\documentclass{sig-alternate}

\usepackage{type1cm}	
\usepackage{graphicx}     
\usepackage{xspace}     
\usepackage{balance}     
\usepackage{booktabs}     
\usepackage[bf,tableposition=top]{caption}     
\usepackage[hyphens]{url}     
\usepackage[bookmarks, pdftex, colorlinks=false]{hyperref}     
\usepackage[square,numbers]{natbib}     
\usepackage{rotating}

\newcommand{\spara}[1]{\smallskip\noindent{\bf #1}}

\newenvironment {squishlist}
{\begin{list}{$\bullet$}
  { \setlength{\itemsep}{0pt}
     \setlength{\parsep}{3pt}
     \setlength{\topsep}{3pt}
     \setlength{\partopsep}{0pt}
     \setlength{\leftmargin}{1.5em}
     \setlength{\labelwidth}{1em}
     \setlength{\labelsep}{0.5em} } }
{\end{list}}

\usepackage[usenames,dvipsnames]{color}
\usepackage{siunitx}	
\usepackage{epsfig}
\usepackage{framed}

\usepackage{amsmath}
\usepackage{amssymb}
\usepackage{pifont}
\usepackage[vlined,linesnumbered,ruled,noend]{algorithm2e}
\usepackage{lscape}
\usepackage{subfigure}
\usepackage{makeidx}
\usepackage{array}
\usepackage{pifont}

\DontPrintSemicolon

\newtheorem{theorem}{Theorem}[section]

\newtheorem{proposition}[theorem]{Proposition}
\newtheorem{definition}[theorem]{Definition}

\newcommand{\cmark}{\text{\ding{51}}}%
\newcommand{\xmark}{\text{\ding{55}}}%

\SetKwInput{Init}{Initialization}
\SetKwInput{BFS}{BFS Traversal}
\SetKwInput{BFSul}{BFS Traversal from $u_L$}
\SetKwInput{Dep}{Dependency Accumulation}
\SetKwInput{Upds}{Data Structures Update}
\SetKw{AndKw}{and}
\SetKw{OrKw}{or}
\SetKw{Continue}{continue}

\newcommand{\spdag}{\textsc{SPdag}\xspace}

\newcommand{\ignore}[1]{}

\usepackage{listings}
\usepackage{multirow}

\newcommand{\dataset}[1]{\textsf{#1}}

\begin{document}

\title{Scalable Online Betweenness Centrality\\in Evolving Graphs}

\numberofauthors{3}
\author{
\alignauthor
Nicolas Kourtellis\\
       \affaddr{Yahoo Labs}\\
       \affaddr{Barcelona, Spain}\\[-4pt]
       \email{\sf kourtell@yahoo-inc.com}
\end{tabular}\begin{tabular}[t]{p{1.3\auwidth}}\centering
Gianmarco De Francisci Morales\\
       \affaddr{Yahoo Labs}\\
       \affaddr{Barcelona, Spain}\\[-4pt]
       \email{\sf gdfm@yahoo-inc.com}
\alignauthor
Francesco Bonchi\\
       \affaddr{Yahoo Labs}\\
       \affaddr{Barcelona, Spain}\\[-4pt]
       \email{\sf bonchi@yahoo-inc.com}
}

\maketitle

\begin{abstract}
Betweenness centrality is a classic measure that quantifies the importance of a graph element (vertex or edge) according to the fraction of shortest paths passing through it.
This measure is notoriously expensive to compute, and the best known algorithm runs in $\mathcal{O}(nm)$ time.
The problems of efficiency and scalability are exacerbated in a dynamic setting, where the input is an evolving graph seen edge by edge, and the goal is to keep the betweenness centrality up to date.
In this paper we propose the first truly scalable algorithm for online computation of betweenness centrality of both vertices and edges in an evolving graph where new edges are added and existing edges are removed.
Our algorithm is carefully engineered with out-of-core techniques and tailored for modern parallel stream processing engines that run on clusters of shared-nothing commodity hardware.
Hence, it is amenable to real-world deployment.
We experiment on graphs that are two orders of magnitude larger than previous studies.
Our method is able to keep the betweenness centrality measures up to date online, i.e., the time to update the measures is smaller than the inter-arrival time between two consecutive updates.
\end{abstract}




\section{Introduction}\label{sec:introduction}

Betweenness centrality measures the importance of an element of a graph, either a vertex or an edge, by the fraction of shortest paths that pass through it~\cite{anthonisse71rush-digraphs, freeman77betweenness, girvan02community}.

Intuitively, an edge that connects two vertices that have many common neighbors is to some extent redundant.
It belongs to a dense area of the graph, and information would be able to propagate even without it.
In other terms, not many shortest paths will need such edge.
This edge is what sociologists call a \emph{strong tie}.

Conversely, a \emph{weak tie} is an edge that connects two vertices with few common neighbors~\cite{Gran73weakties}.
Such an edge is likely to bridge two distinct dense areas of the graph (also called \emph{communities}), hence it participates in many shortest paths.
Information has to go over this ``bridge'' to propagate from one community to another.
The two vertices that own the bridge have a strategically favorable position because they can block information, or access it before the other individuals in their community: they span a \emph{``structural hole'}'~\cite{Burt92structuralholes}.

\citet{girvan02community} exploit this concept to define one of the first and most elegant algorithms for \emph{community detection}.
The algorithm iteratively removes the highest betweenness edge and produces a hierarchical decomposition of the graph, where the remaining disconnected components are the communities discovered.

Betweenness centrality has been used to analyze social networks~\cite{kahng03betweenness-correlation,liljeros01sexual-contacts}, protein networks~\cite{jeong01protein-nets}, wireless ad-hoc networks~\cite{maglaras11adhoc}, mobile phone call networks~\cite{catanese12forensic}, multiplayer online gaming networks~\cite{ang11mmog}, to inform the design of socially-aware P2P systems~\cite{kourtellis13p2pcentrality}, just to name a few examples.

Measuring betweenness centrality requires computing the shortest paths between all pairs of vertices in a graph.
This computation is possible in small graphs with a few tens of thousands vertices and edges, but it quickly becomes prohibitively expensive as the graphs grow larger.
Indeed, the best known algorithm for betweenness centrality,  proposed by \citet{brandes01betweenness}, runs in $\mathcal{O}(nm)$ time.

Due to its high cost, some works have proposed to parallelize its execution~\cite{bader06parallel-betweenness}.
Others, have proposed to approximate betweenness centrality through the use of randomized algorithms~\cite{brandes07centralityestimate,kourtellis12kpath,riondato2014approxbetweenness}.
However, the accuracy of these randomized algorithms can decrease considerably with the increase in graph size~\cite{kourtellis12kpath}.
Variants of betweenness centrality, such as flow betweenness~\cite{freeman91flow-betweenness} and random-walk betweenness~\cite{newman05random-walk-betweenness}, also run in $\mathcal{O}(nm)$ time. Finally, there are no cheaper measures that can be used as a proxy, as they do not to correlate well with betweenness centrality~\cite{Boldi13axiomscentrality} (differently from, e.g., degree centrality for PageRank).
So there is no easy workaround to the complexity of computing betweenness centrality.
State-of-the-art methods  are too expensive for graphs with millions of vertices and edges, thus making this measure hard to use in practical application scenarios.

The picture gets even worse when considering that real graphs of interest, such as the Web, social networks, and information networks, are dynamic in nature and evolve continuously with new edges and vertices arriving, and old edges being removed.
In such a scenario, the na\"{i}ve approach of recomputing the measure from scratch is impractical even on moderately large graphs.

\subsection{Related work}\label{sec:rel-work}

Recently there have been three main proposals by~\citet{lee12qube},~\citet{green12iter-betweenness} and~\citet{kas13betweenness} for the incremental computation of vertex betweenness centrality.

QUBE~\cite{lee12qube} relies on the decomposition of the graph into disjoint minimum union cycles (MUCs).
The algorithm uses this decomposition to identify vertices whose centrality can potentially change.
If the updates to the graph do not affect the decomposition, then the centrality needs to be recomputed (from scratch) only within the MUCs of the affected vertices.
If the decomposition changes, a new one must be computed, and then new centralities must be computed for all the affected components.
The performance of the algorithm is tightly connected to the size of the MUCs found, which in real (and especially social) graphs can be very large.
The algorithm depends on a preprocessing step that computes a minimum cycle basis for the given graph.
Then a MUC decomposition is computed by recursively taking the union of cycles in the minimum cycle basis that share a vertex.
Therefore, it requires $\mathcal{O}(m)$ storage.

QUBE leverages Brandes'~\cite{brandes01betweenness} algorithm to compute the centrality scores inside a MUC.
However, the algorithm also requires to compute the (number of) shortest paths between each pair of vertices in the affected MUC, which could require $\mathcal{O}(n^2)$ space in the worst case.
The overall space complexity is thus $\mathcal{O}(n^2+m)$.

Green et al.~\cite{green12iter-betweenness} propose to maintain the previously computed values of betweenness and the data structures needed by Brandes' algorithm~\cite{brandes01betweenness}, and update the ones affected by a graph change.
Their approach has a space complexity  of $\mathcal{O}(n(n+m))$, which becomes prohibitive for large graphs of millions of vertices.

Kas et al.~\cite{kas13betweenness} extend the work proposed by~\citet{ramalingam92incremental} to accommodate the computation of vertex betweenness centrality while adding edges or vertices.
Differently from Brandes' algorithm, their technique does not use dependencies but the actual shortest distances.
It keeps the data structures needed to update the betweenness centrality of vertices: distances, number of shortest paths and predecessors list.
The computational complexity can be at most as Brandes', i.e., $\mathcal{O}(nm)$.
However, the space complexity is the same as~\citet{green12iter-betweenness}, i.e., $\mathcal{O}(n(n+m))$.

Recently, and concurrently with our study, some works have also improved the space complexity of the incremental betweenness computation to $O(n^2)$~\cite{mclaughlin14gpu-betweenness, nasre14mfcs-betweenness}.
Similarly to our approach,~\citet{mclaughlin14gpu-betweenness} modifies the technique by~\citet{green12iter-betweenness} by dropping the predecessor lists, and porting it to GPUs to run on larger graphs.
However, it has only been tested on a small sample of source nodes (256) instead of the full graph.
Thus, its scalability to large graphs with higher memory demands, as well as its capability to follow dynamic rates of graph updates have not been demonstrated.
\citet{nasre14mfcs-betweenness} present a tighter upper bound of $O(nm^*)$ on time complexity, where $m^*$ is the number of edges that lie on shortest paths.
This is the first algorithm to have a lower time complexity than Brandes'~\citep{brandes01betweenness}.
However, while usually smaller, $m^*$ is still $O(m)$ in the worst case (e.g., an unweighted clique).
Most importantly, their algorithm is not straightforward to parallelize, as it requires access to the shortest paths \textsc{DAG} (\spdag) rooted in $v$ from each \spdag rooted in $s$ ($\forall s \in V$), where $v$ is an endpoint of the updated edge.
Therefore, it is not possible to parallelize the algorithm by distributing the set of \spdag on several machines, as in our approach.
Moreover, neither of these approaches can handle removal of edges or provide updated edge betweenness scores.

\begin{table}[t!]
\caption{Comparison with previous studies:  vertex ($C_V$) and edge centrality  ($C_E$), edge addition (+) and removal (-), parallel and streaming computation ($\|$), size of the largest graph used in the experiments ($|V|$ and $|E|$).
Note that \citet{nasre14mfcs-betweenness} have smaller time complexity than Brandes' and other algorithms.}
\centering
\small
\tabcolsep=0.06cm
\begin{tabular}{ccccccccrr}
\toprule
Method							& Year 	&	Space				&	$C_V$	&	$C_E$	&	$+$	&	$-$	&	$\|$	&	$|V|$ & $|E|$	\\
\midrule
\citet{lee12qube} 					& 2012	&	$\mathcal{O}(n^2$+$m)$	&	$\cmark$	&	$\xmark$	&	$\cmark$	&	$\cmark$	&	$\xmark$		&	12k	&	65k 		\\
\citet{green12iter-betweenness}		& 2012 	&	$\mathcal{O}(n^2$+$nm)$	&	$\cmark$	&	$\xmark$	&	$\cmark$	&	$\xmark$	&	$\xmark$		&	23k	&	94k		\\
\citet{kas13betweenness}				& 2013	&	$\mathcal{O}(n^2$+$nm)$	&	$\cmark$	&	$\xmark$	&	$\cmark$	&	$\xmark$	&	$\xmark$		&	8k	&	19k		\\
\citet{nasre14mfcs-betweenness}				& 2014	&	$\mathcal{O}(n^2)$	&	$\cmark$	&	$\xmark$	&	$\cmark$	&	$\xmark$	&	$\xmark$		&	-	&	-		\\
This work							& 2014	&	$\mathcal{O}(n^2)$		&	$\cmark$	&	$\cmark$	&	$\cmark$	&	$\cmark$	&	$\cmark$		&	2.2M	&	5.7M		\\
\bottomrule
\end{tabular}
\label{tab:rel-work}
\end{table}

\subsection{Contributions}

The main contribution of this paper is to provide \emph{the first truly scalable and practical framework for computing vertex and edge betweenness centrality of large evolving graphs, incrementally and online}.
Our proposal represents an advancement over the state of the art in four main aspects as summarized in Table~\ref{tab:rel-work}.

First, our method maintains \textit{both} vertex and edge betweenness centrality up-to-date for the same computational cost, while the previously proposed methods are only tailored for vertex betweenness.

Second, we handle \textit{both additions and removals} of edges in a unified approach, while the previously proposed methods, besides QUBE, can handle only addition of edges.
In fact, we show that the incremental, up-to-date edge betweenness under continuous edge removals allows for faster execution of the Girvan-Newman~\cite{girvan02community} algorithm on larger graphs.

Third, our method has reduced space overhead.
Similarly to previous work~\cite{green12iter-betweenness}, our algorithm maintains the previously computed values of betweenness and other needed data structures, and updates the ones affected by graph changes.
However, our method avoids maintaining the predecessors lists, thus reducing the space complexity to $\mathcal{O}(n^2)$~\cite{green13datastructurebetweenness}.
This optimization requires a scan of all neighbors instead of predecessors: we show that this does not affect the time complexity, and makes the algorithm more scalable and faster in practice.

Forth, our framework is truly scalable and amenable to real-world deployment.
The framework is carefully engineered to use \emph{out-of-core} techniques to store its data structures on disk in a compact binary format.
Data structures are read sequentially by employing \emph{columnar storage}, and memory structures are \emph{mapped directly on disk} to minimise memory copies.

Finally, we show how our method can be parallelized and deployed on top of modern parallel data processing engines that run on clusters of commodity hardware, such as Storm\footnote{\url{http://storm.apache.org}}, S4\footnote{\url{http://incubator.apache.org/s4}}, Samza\footnote{\url{http://samza.apache.org}}, or Hadoop.
Our experiments test our method on graphs with millions of vertices and edges, i.e., two orders of magnitude larger than previous studies.
By experimenting with real-world evolving graphs, we also show that our algorithm is able to keep the betweenness centrality measure up to date \emph{online}, i.e., the time to update the measure is always smaller than the inter-arrival time between two consecutive updates.
\enlargethispage*{2\baselineskip}

An open-source implementation of our method is available on GitHub.\footnote{\url{http://github.com/nicolas-kourtellis/StreamingBetweenness}}
Inclusion in SAMOA,\footnote{\url{http://samoa.incubator.apache.org}} a platform for mining big data streams~\cite{deFrancisciMorales2013samoa, deFrancisciMorales2015samoa} is also planned.

\spara{Roadmap.} Section~\ref{sec:brandes} introduces the basic concepts and recalls Brandes' algorithm.
We overview our algorithm for online betweenness centrality in Section~\ref{sec:overview} and cover the details on insertions and removals in Section~\ref{sec:add-rem}.
Section~\ref{sec:scalable} presents various optimizations that make our method scalable to large graphs.
Section~\ref{sec:exp-methodology-results} reports our experimental results, while
Section~\ref{sec:discussion} concludes the paper.

\section{Preliminaries}\label{sec:brandes}

Let $G = (V,E)$ be a (directed or undirected) graph, with $|V|=n$ and $|E| = m$.
Let $P_{s}(t)$ denote the \emph{set of predecessors} of a vertex $t$ on shortest paths from $s$ to $t$ in $G$.
Let $\sigma(s,t)$ denote the \emph{total number of shortest paths} from $s$ to $t$ in $G$ and, for any $v \in V$, let $\sigma(s,t \mid v)$ denote the number of shortest paths from $s$ to $t$ in $G$ that \emph{go through $v$}.
Note that $\sigma(s,s) = 1$, and $\sigma(s,t \mid v) = 0$ if $v \in \{s, t\}$ or if $v$ does not lie on any shortest path from $s$ to $t$.
Similarly, for any edge $e \in E$, let $\sigma(s,t \mid e)$ denote the number of shortest paths from $s$ to $t$ in $G$ that \emph{go through $e$}.
The betweenness centrality of a vertex $v$ is the sum over all pairs of vertices of the fractional count of shortest paths going through $v$.
\begin{definition}[Vertex Betweenness Centrality]\label{def:bet-cent-vertex-def}
For every vertex $v \in V$ of a graph $G(V, E)$, its \textup{betweenness centrality} $VBC(v)$ is defined as follows:
\begin{equation}
VBC(v) = \sum_{s,t \in V, s \neq t}  \frac{\sigma(s,t \mid v)}{\sigma(s,t)}.
\end{equation}
\end{definition}

\begin{definition}[Edge Betweenness Centrality]\label{def:bet-cent-edge-def}
For every edge $e \in E$ of a graph $G(V, E)$, its \textup{betweenness centrality} $EBC(e)$ is defined as follows:
\begin{equation}
EBC(e) = \sum_{s,t \in V, s \neq t} \frac{\sigma(s,t \mid e)}{\sigma(s,t)}.
\end{equation}
\end{definition}

\spara{Brandes' algorithm~\cite{brandes01betweenness}} leverages the notion of  \emph{dependency score} of a source vertex $s$ on another vertex $v$, defined as $\delta_{s }(v) = \sum_{t \neq s, v} \frac{\sigma(s,t \mid v)}{\sigma(s,t)}$.
The betweenness centrality $VBC(v)$ of any vertex $v$ can be expressed in terms of dependency scores as $VBC(v) = \sum_{s \neq v} \delta_{s}(v)$.
The following recursive relation on $\delta_{s}(v)$ is the key to Brandes' algorithm:
\begin{equation}\label{eq:dependency}
\delta_{s}(v) = \sum_{w:v \in P_{s}(w)}\frac{\sigma(s,v)}{\sigma(s,w)}(1+\delta_{s}(w))
\end{equation}

The algorithm takes as input a graph $G$=$(V,E)$ and outputs the betweenness centrality $VBC(v)$ of every $v \in V$.
It runs in two phases.
During the first phase, it performs a search on the whole graph to discover shortest paths, starting from every source vertex $s$.
When the search ends, it performs a \emph{dependency accumulation} step by backtracking along the shortest paths discovered.
During these two phases, the algorithm maintains four data structures for each vertex found on the way: a predecessors list $P_s[v]$, the distance $d_s[v]$ from the source, the number of shortest paths from the source $\sigma_s[v]$, and the dependency $\delta_s[v]$ accumulated when backtracking at the end of the search.

On unweighted graphs, Brandes'  algorithm uses a breadth first search (BFS) to discover shortest paths, and its running time is $\mathcal{O}(nm)$.
The space complexity of the algorithm is $\mathcal{O}(m+n)$.
While this algorithm was initially defined only for vertex betweenness it can be easily modified to produce edge betweenness centrality at the same time~\cite{brandes08variants-betweenness}.

\section{Framework Overview}\label{sec:overview}

Our framework computes betweenness centrality in evolving unweighted graphs.
We assume new edges are added to the graph or existing edges are removed from the graph, and these changes are seen as a stream of updates, i.e., one by one.
Henceforth, for sake of clarity, we assume an undirected graph.
However, our framework can also work on directed graphs by following outlinks in the search phase and inlinks in the backtracking phase rather than generic neighbours.

The framework is composed of two basic steps shown in Figure~\ref{fig:framework-addition-deletion}.
It accepts as input a graph $G(V,E)$ and a stream of edges $E_S$ to be added/removed, and outputs, for an updated graph $G'(V',E')$, the new betweenness centrality of vertices ($VBC'$) and edges ($EBC'$) for each vertex $v \in V'$ and edge $e \in E'$.

\begin{figure}[t]
\begin{framed}
\vspace{-2mm}
\small{
	\KwIn{Graph $G(V,E)$ and edge update stream $E_S$}
	\KwOut{$VBC'[V']$ and $EBC'[E']$ for updated $G'(V',E')$}
	\begin{description}
		\item [Step $1$:] Execute Brandes' alg. on G to create \& store data structures for incremental betweenness.
		\item [Step $2$:] \textbf{For each} update $e$$\in$$E_S$, execute Algorithm~\ref{algo:iterative-overview-addition-deletion-neighbors}.
		\Indp
		\item [Step $2.1$] $\text{  }$ Update vertex and edge betweenness.
		\item [Step $2.2$] $\text{  }$ Update data structures in memory or disk for next edge addition or removal.
		
	\end{description}
}
\vspace{-4mm}
\end{framed}
\caption{The proposed algorithmic framework.}\label{fig:framework-addition-deletion}
\end{figure}

The framework uses Brandes' algorithm as a building block in step 1:
this is executed only once, offline, before any update.
We modify the algorithm to ($i$) keep track of betweenness for vertices and edges at the same time, ($ii$) use additional data structures to allow for incremental computation, and ($iii$) remove the predecessors list to reduce the memory footprint and make out-of-core computation efficient.

\spara{Edge betweenness.}
By leveraging ideas from~\citet{brandes08variants-betweenness}, we modify the algorithm to produce edge betweenness centrality scores.
To compute simultaneously both edge and vertex betweenness, the algorithm stores the intermediate dependency values (Eq.~\ref{eq:dependency}) independently for each vertex.

\spara{Additional data structures.}
To allow for incremental computation we need to maintain some additional data.
In particular, we need a compact representation of the \emph{directed acyclic graph of shortest paths rooted in the source vertex} (which we refer to as \spdag), and the accumulated dependency values.
Thus, for each source vertex $s$ we maintain an additional data structure $BD[s]$ that stores its \emph{betweenness data}.
$BD[s]$ stores three pieces of information for each other vertex $t$:
\begin{squishlist}
\item $BD[s].d[t]$: the distance of vertex $t$ from source $s$;
\item $BD[s].\sigma[t]$: the number of shortest paths starting from source $s$ and ending at the given vertex $t$;
\item $BD[s].\delta[t]$: the dependency accumulated on the vertex $t$ in the backtracking to source $s$.
\end{squishlist}
The data structure is initialized in step 1, and populated at the end of the dependency accumulation phase.
Then, it is used in step 2.1 and updated in step 2.2.

\spara{Memory optimisation.}
Brandes' algorithm builds a list of predecessors during the search phase to speed up the backtracking phase.
Differently from the other data structures, the size of this list is variable and can grow considerably.
For example, by assuming just \num{4} predecessors on average, the size of the list would be as large as $BD[\cdot]$ (assuming integer identifiers).

To reduce the space complexity of the algorithm, we remove the predecessors lists.
When backtracking in the dependency accumulation phase, the algorithm checks all the neighbors of the current vertex, and uses the level of the vertex in the \spdag (i.e., the distance from the source) to pick the next vertices to visit.

For each source, we need to maintain an \spdag to all other vertices in the graph ($\mathcal{O}(n)$), along with the edges to their predecessors ($\mathcal{O}(m)$).
Therefore, each \spdag takes $\mathcal{O}(n+m)$ space.
In total, the space complexity of the original algorithm is $\mathcal{O}(n(n+m))$ with the predecessors lists.
By removing the predecessors lists, we reduce the space complexity to $\mathcal{O}(n^2)$.
Furthermore, the time complexity remains unchanged, as shown next.

To understand why removing the predecessors lists does not increase the time complexity, consider that the order in which vertices are traversed is unchanged.
Assume there are $k$ edges in the predecessors lists overall.
Brandes' algorithm checks each edge once during the search phase to populate the predecessors lists.
In the backtracking phase only the edges in predecessors list are checked.
Thus, overall the complexity of the algorithm is $O(m+k)$.
In the worst case, $k$ is of the same order as $m$, hence, the complexity of traversing all the predecessors lists is bounded by $\mathcal{O}(m)$.
Scanning the neighbors of each vertex is also bounded by $\mathcal{O}(m)$.
Thus, the algorithm's worst-case time complexity is unchanged.

An additional benefit of removing the predecessors lists is to also avoid the overhead of building it during the traversal.
In practice, this optimization not only reduces the space complexity, but also decreases the average running time of the algorithm, as shown in Section~\ref{sec:exp-methodology-results} and in previous work~\cite{green13datastructurebetweenness,sariyuce13betweennessGPUs}.
Moreover, by removing the predecessors lists, we do not need to maintain any variable-length data structure in our algorithm.
This simplification allows us to use very efficient out-of-core techniques to manage $BD[\cdot]$ when its size outgrows the available main memory (see Section \ref{subsec:out-of-core}).

\begin{algorithm}[t]
\small{	
	\KwIn{$G(V,E)$, $(u_1,u_2)$, $VBC[V]$, $EBC[E]$, $BD[V]$}
	\KwOut{$VBC'[V']$, $EBC'[E']$, $BD'[V']$}
	\Init{
		$EBC'[(u_1,u_2)] = 0.0$\;
		\textbf{Addition:} $E' \leftarrow E \cup (u_1,u_2);$
		\textbf{Deletion:} $E' \leftarrow E \setminus (u_1,u_2)$\;
	}
		
	\For{$s \in V'$}{
		$u_L \leftarrow$ findLowest($u_1$, $u_2$, $BD[s]$)\;
		$u_H \leftarrow$ findHighest($u_1$, $u_2$, $BD[s]$)\;
		
		$dd \leftarrow BD[s].d[u_L] - BD[s].d[u_H]$\;
		\If{$dd == 0$}{
			\Continue \tcp{same level addition/deletion}
		}
		\If{$dd \geq 1$}{
			\For{$r \in V$}{
				$\sigma'[r]=BD[s].\sigma[r];$
				$d'[r]=BD[s].d[r];$
				$\delta'[r]=0$\;
				$t[r] \leftarrow N_T$ \tcp{not touched before}
			}

			$LQ[|V|] \leftarrow$ empty queues;
			$Q_{BFS} \leftarrow$ empty queue\;
			$Q_{BFS} \leftarrow u_L$\;
			
			\If{addition}{
				\If{$dd==1$}{
					execute Alg.~\ref{algo:L1-addition-removal-neighbors}; \tcp{0 level rise}
				}
				\If{$dd > 1$}{
					execute Alg.~\ref{algo:L2-addition-neighbors}; \tcp{1 or more level rise}
				}
			}
			\ElseIf{deletion}{
				\If{$u_L$ has predecessors}{
					execute Alg.~\ref{algo:L1-addition-removal-neighbors}; \tcp{0 level drop}
				}
				\textbf{else}
					execute Alg.~\ref{algo:L2+removal-neighbors-pivots} \tcp{1 or more level drop}
			}
		}
		
		\Upds{}
		\For{$r \in V'$}{
			$BD[s].\sigma[r]=\sigma'[r]$;
			$BD[s].d[r]=d'[r]$\;
			\textbf{if $t[r]\neq N_T$ then}
				$BD[s].\delta[r]=\delta'[r]$\;
		}
	}
}
	\caption{Incremental computation of vertex and edge betweenness when adding or removing an edge.}\label{algo:iterative-overview-addition-deletion-neighbors}
\end{algorithm}

\subsection{Addition and removal of edges}
The addition of a new edge may cause structural changes in a graph, i.e., changes in the distance between vertices.
Depending on the previous distance between the two newly-connected endpoints, these changes may bring some vertices closer to the current source $s$.
Similarly, when an existing edge is removed structural changes may move the furthest endpoint away from the source.

\spara{Edge addition.} In step~2 of the framework the Algorithm~\ref{algo:iterative-overview-addition-deletion-neighbors} is executed to update the betweenness centrality of the graph when a new edge $(u_1,u_2)$ is added.
The algorithm has access to the data structure $BD[\cdot]$ computed in step 1, and runs independently for each source.
For a given source $s$, the algorithm uses $BD[s]$ to determine which endpoint of the new edge is closest to the current source $s$ (denoted $u_H$) and which one is furthest (denoted $u_L$). Let $dd(u_L,u_H)$ denote the difference in distance from the source of the two endpoints before the addition, i.e., $dd(u_L,u_H) = d(s,u_L) - d(s,u_H)$ where $d$ represents shortest path distance. Since only one edge is added at a time, we simply denote this difference as $dd$.

Depending on how large $dd$ is, a different type of update is needed.
In particular, three cases can arise:
\begin{squishlist}
\item $dd = 0$ (Proposition \ref{prop:same-level});
\item $dd = 1$ (\textbf{0 level rise},  Section \ref{subsec:0lir});
\item $dd > 1$ (\textbf{1 or more levels rise},  Section \ref{subsec:1li}).
\end{squishlist}
The first case involves two vertices that are at the same distance from the source vertex.
\begin{proposition}\label{prop:same-level}
Given two vertices $u_1$ and $u_2$ such that they have the same distance from a source vertex $s$, and an edge $e = (u_1,u_2)$ that connects the two vertices, no shortest path from $s$ to any other node in the graph passes trough the edge $e$, i.e., 
$d(s,u_1) = d(s,u_2) \implies \forall t \in V, \sigma(s, t \mid e) = 0$.
\end{proposition}
\begin{proof}
We prove the proposition by contradiction.
Assume there exists a shortest path from vertex $s$ to $p$ that goes through the edge $(u_1,u_2)$, i.e., path $s, \ldots, u_1,u_2, \ldots, p$.
However, since $u_1$ and $u_2$ are at the same distance from the source $s$, we can construct another path that is one hop shorter, that starts from $s$ and ends in $p$ but skips $u_1$, i.e., $s,\ldots, u_2, \ldots, p$, which contradicts the assumption that $s, \ldots, u_1,u_2, \ldots, p$ is a shortest path.
\end{proof}
No shortest path goes through the edge, no change occurs in the \spdag, so the current source can be ignored.

In the second case, the new edge connects two vertices whose distance from the source differs only by one (Fig.~\ref{fig:examples-additions-deletions}a).
Thus, this addition does not cause any structural change in the \spdag, and all the distances remain the same.
However, new shortest paths can be created due to the addition, and therefore the shortest paths and the dependencies of the graph must be updated.

In the third and most complex case, $dd > 1$, structural changes occur in the \spdag (Fig.~\ref{fig:examples-additions-deletions}b depicts this case after the rise of $u_L$).
In order to handle these changes properly, we introduce the concept of \emph{pivot}.
\begin{definition}[Pivot]
Let $s$ be the current source, let $d()$ and $d'()$ be the distance before and after an update, respectively, we define \emph{pivot} a vertex $p_V \mid d(s,p_V) = d'(s,p_V) \wedge \exists \, w \in neighbors(p_V)$$: d(s,w)$$\neq$$d'(s,w)$.
\end{definition}

Thus, a pivot is a vertex that, under an edge addition or removal, does not change its distance from the source $s$, but has neighbors that do so.

When $dd > 1$, we need to first compute the new distances by leveraging the pivots.
Given that their distance has not changed, we can use them as starting points to correct the distances in the \spdag.
In the case of addition, all the pivots are situated in the sub-dag rooted in $u_L$, so we can combine the discovery of the pivot with the correction of the shortest paths.
The different cases that can arise are discussed in detail in Section~\ref{subsec:1li}.

There exists also a fourth case: the new edge connects two previously disconnected components.
This case degenerates into the case $dd=1$.
Indeed, no previous shortest path existed between the two disconnected components, so there is no structural change in the \spdag.

Finally, new vertices arriving in the graph are handled simply by adding them to the source set $V'$ with a zero $VBC'$.
Then, for all sources, the new vertex is considered as $u_L$ with $d[u_L]=d[u_H]+1$, where $u_H$ is the other endpoint of the incoming edge (therefore $dd = 1$).


\spara{Edge removal.}
In the case of an edge ($u_1$, $u_2$) removed from the graph, $dd$ is at most one, as the two endpoints are connected before the removal.
In this case, one of the two endpoints, $u_H$, is closest to the source, and clearly the edge $(u_H, u_L)$ belongs to at least one shortest path from the source $s$ to $u_L$.
Therefore, the algorithm needs to check whether $u_L$ has other shortest paths from $s$, not passing trough $(u_H, u_L)$.
Again, there are three cases:
\begin{squishlist}
\item $dd = 0$ (Proposition \ref{prop:same-level});
\item $dd = 1$ and $u_L$ has other predecessors (\textbf{0 level drop},  Section \ref{subsec:0lir});
\item $dd = 1$ and $u_L$ has no other predecessor (\textbf{1 or more levels drop},  Section \ref{subsec:1lr}).
\end{squishlist}

In the first case there are no shortest paths passing through the edge.
Therefore no changes occur in the \spdag, so the current source $s$ can be skipped.

In the second case, if $u_L$ is connected to at least one vertex $u'_H$ such that $dd(u_H, u'_H) = 0$, then $u_L$ will remain at the same distance (Fig.~\ref{fig:examples-additions-deletions}a), and no structural change occurs.
Thus distances remain the same.
However, some shortest paths coming through $(u_H, u_L)$ are lost, so the betweenness centrality needs to be updated.

In the third and most complex case, structural changes occur in the graph (Fig.~\ref{fig:examples-additions-deletions}b depicts this case before $u_L$ drops).
Also in this case we make use of pivots to correct the distances in the \spdag first, and subsequently adjust the shortest paths and dependency values.
However, not all pivots will be found in the sub-dag rooted in $u_L$ after the removal.
This difference makes this case more complicated than the addition, as some pivots cannot be discovered while adjusting the shortest paths (e.g., if nodes $u_L$ and $r$ were connected).
Therefore, we need to first search and find the pivots, and then start a second BFS from those pivots to correct the shortest paths.
The details of this case are covered in Section~\ref{subsec:1lr}.

There is also the case where the edge removed disconnects the sub-dag rooted in $u_L$ from the rest of the graph (or, similarly, turns $u_L$ into a singleton).
In this case, the shortest paths coming from the source, as well as the dependencies going to the source from this component must be removed and the betweenness adjusted.
If $u_L$ is to be removed, all its edges are iteratively removed and the singleton is replaced with zero VBC' (Section~\ref{subsec:discon}).

\begin{figure}[]
\begin{center}
	\includegraphics[scale=1.0]{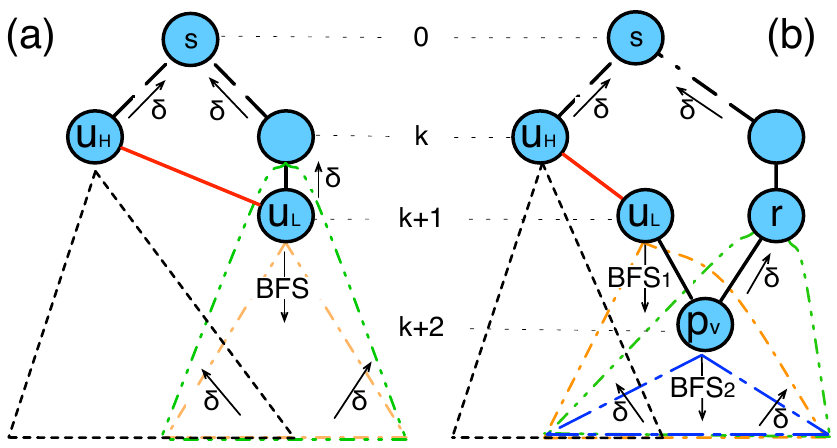}
\caption{The red (light) edge is added/removed and either does not cause structural changes (a), or does so (b).}
\label{fig:examples-additions-deletions}
\end{center}
\end{figure}

\section{Incremental Addition \& Removal}\label{sec:add-rem}

In this section we discuss the details of our framework in the case of edge addition and removal.

\subsection{No level change}\label{subsec:0lir}

\begin{algorithm}[t]
\small{
	\BFSul{}
	$LQ[d'[u_L]] \leftarrow u_L$; $t[u_L] \leftarrow D_N$\;
	\textbf{If $addition$ then} $\sigma'[u_L]$+=$BD[s].\sigma[u_H];$\;
	\textbf{If $deletion$ then} $\sigma'[u_L]$-=$BD[s].\sigma[u_H];$\;

	\While{$Q_{BFS}$ not empty}{
		$v \leftarrow Q_{BFS}$\;
		\For{$ w \in neighbors(v)$}{
			\If{$d'[w]==d'[v]+1$}{
				\If{$t[w]==N_T$}{
					$t[w]\leftarrow D_N;$
					$Q_{BFS} \leftarrow w;$
					$LQ[d'[w]]$$\leftarrow$$w$\;
				}
				$\sigma'[w]+=\sigma'[v]-BD[s].\sigma[v]$\;
			}
		}
	}

	\Dep{}
	\If{$deletion$}{
		$\delta'[u_H]=BD[s].\delta[u_H]-\frac{\sigma[u_H]}{\sigma[u_L]}(1+BD[s].\delta[u_L])$\;
		$LQ[d'[u_H]] \leftarrow u_H;$
		$t[u_H]\leftarrow U_P;$
	}
	$level=|V'|$\;
	\While{$level>0$}{
		\While{$LQ[level]$ not empty}{
			$w \leftarrow LQ[level]$\;
			\For{$ v \in neighbors(w)$}{
				\If{$d'[v] < d'[w]$}{

					Execute module in Alg.~\ref{algo:dep-mod-1}.

					\If{$addition$}{
						\If{$t[v]==U_P$ \AndKw $(v \neq u_H$ \OrKw $w \neq u_L)$}{
							$\delta'[v]-=\alpha$\;
						}
						\textbf{if $(v,w) \neq (u_L,u_H)$ then}
							$EBC'[(v,w)]-= \alpha$\;
					}
					\If{$deletion$}{
						$EBC'[(v,w)]-= \alpha$\;
						\textbf{if $t[v]==U_P$ then}
							$\delta'[v]-=\alpha$\;
					}
				}
			}
			\textbf{if $w\neq s$ then}
				$VBC[w]+=\delta'[w]-BD[s].\delta[w]$\;
		}
		$level = level - 1;$
	}
	}
	\caption{Betweenness update for addition or removal of an edge where $u_L$ remains at the same level after the update.}\label{algo:L1-addition-removal-neighbors}
\end{algorithm}

\begin{algorithm}[t]
\small{
	\If{$t[v]==N_T$}{
		$t[v]\leftarrow U_P;$ $\delta'[v]=BD[s].\delta[v];$ $LQ[level$-$1]$$\leftarrow$$v$\;
	}
	$c=\frac{\sigma'[v]}{\sigma'[w]}(1+\delta'[w]);$ $\delta'[v]$+=$c;$
	$EBC'[(v,w)]+=c$\;
	$\alpha=\frac{BD[s].\sigma[v]}{BD[s].\sigma[w]}(1+BD[s].\delta[w])$\;
}
\caption{Initialization dependency module.}\label{algo:dep-mod-1}
\end{algorithm}

Algorithm~\ref{algo:L1-addition-removal-neighbors} handles the case when the edge added or removed does not cause structural changes in the \spdag, i.e., $d(s, u_L)$ is unchanged (Fig.~\ref{fig:examples-additions-deletions}a).
The algorithm initializes $u_L$ by adding or removing all shortest paths from its predecessor $u_H$, depending on whether the edge $(u_H,u_L)$ is added or removed.
It also maintains a state flag $t[v]$ for each vertex $v$ which represents the direction in which the algorithm has encountered the vertex $v$: $D_N$ if descending (searching), $U_P$ if ascending (backtracking), and $N_T$ if untouched.

The vertices whose shortest paths from a source $s$ have potentially been altered are situated in the sub-dag rooted in $u_{L}$.
Thus, Algorithm~\ref{algo:L1-addition-removal-neighbors} performs a BFS traversal of the \spdag starting from $u_L$ by visiting neighbors whose distance from $s$ is higher than the current vertex (line 7, orange single-dotted triangle in Fig.~\ref{fig:examples-additions-deletions}a).
During the BFS, it updates the number of shortest paths to each vertex found to take into account the paths created (lost) by the addition (removal) of the edge $(u_H, u_L)$.
To avoid double counting, the old number of shortest paths from each predecessor $BD[s].\sigma[v]$ is subtracted before adding the new number of shortest paths $\sigma'[v]$ (line 10).
The vertices encountered in the BFS are added to a level-specific queue $LQ$ for later use.
At the end of the BFS, the number of shortest paths is up-to-date.

In the dependency accumulation phase, the algorithm polls the $LQ$ queue for each level and visits the vertices in reverse order of discovery by the BFS ($\delta$ arrows in Fig.~\ref{fig:examples-additions-deletions}a).
In the case of edge removal (lines 11-13), $u_H$ is inserted in $LQ$ before the dependency accumulation starts, to guarantee that the data structures of $u_H$ and its predecessors will be updated even though the edge ($u_H$,~$u_L$) does not exist anymore.

The vertices whose dependency needs an update are situated in two parts of the \spdag.
Part of them is in the sub-dag rooted in $u_L$, as discovered by the BFS.
For these vertices, the number of shortest paths was updated during the BFS.
They have been inserted in the appropriate queue $LQ$ and will be examined during the dependency accumulation phase.
The others are the predecessors of the sub-dag, which lay at the fringe of the first ones (green, double-dotted triangle in Fig.~\ref{fig:examples-additions-deletions}a).
These vertices are found while backtracking in the dependency accumulation phase by examining the neighbors of vertices in $LQ$.
Discovered predecessors are added in $LQ$ and examined in the next level (line 2 in Alg.~\ref{algo:dep-mod-1}).

During the backtracking phase, the algorithm corrects the dependency of vertices that are affected.
If a vertex was not encountered before, the old dependency for the particular path is subtracted (line 23 for addition, line 27 for removal).
The new dependency $c$ is added in line 3 of Alg.~\ref{algo:dep-mod-1} and then corrected by subtracting the old dependency $\alpha$ in line 24 for addition and line 26 for removal.
Similarly the $VBC$ is updated in line 28.

\subsection{Addition: 1 or more levels rise}\label{subsec:1li}

In this case, the \spdag undergoes structural changes.
To handle these changes, we need to find the pivots, which define the boundary of the structural change of the \spdag.
Given that the addition of the edge can only bring vertices closer to the source if they are reachable from $u_L$, all the pivots can be found while traversing the \spdag with a BFS starting from $u_L$, similarly to the previous case (orange, single-dotted triangle in Fig.~\ref{fig:examples-additions-deletions}b).

\begin{algorithm}[t]
\small{
	\BFSul{}
	$LQ[d[u_L]] \leftarrow u_L;$
	$d'[u_L]=BD[s].d[u_H]+1$\;

	\While{$Q_{BFS}$ not empty}{
		$v \leftarrow Q_{BFS};$	$t[v]\leftarrow D_N;$	$\sigma'[v]=0$\;

		\For{$ w \in neighbors(v)$}{
			\textbf{if $d'[w]+1==d'[v]$ then}
				$\sigma'[v]+=\sigma'[w];$\;
			\If{$d'[w]>d'[v]$ $\AndKw$ $t[w]==N_T$}{
				$t[w]$$\leftarrow$$D_N;$
				$d'[w]$=$d'[v]$+$1;$
				$LQ[d'[w]]$$\leftarrow$$w;$
				$Q_{BFS}$$\leftarrow$$w$\;
			}
			\If{$d'[w]==d'[v]$ \AndKw $BD[s].d[w] \neq BD[s].d[v]$}{
				\If{$t[w]==N_T$}{
					$t[w]\leftarrow D_N;$
					$LQ[d'[w]] \leftarrow w;$
					$Q_{BFS} \leftarrow w$
				}
			}
		}
	}
	\Dep{}
	$level=|V'|;$
	$te[e]\leftarrow N_T, e \in E$\;
	\While{$level>0$}{
		\While{$LQ[level]$ not empty}{
			$w \leftarrow LQ[level]$\;
			\For{$ v \in neighbors(w)$}{
				\If{$d'[v] < d'[w]$}{
					Execute module in Alg.~\ref{algo:dep-mod-1}.
				
					\If{$(t[v]=U_P)$ \AndKw $(v \neq u_H$ \OrKw $w \neq u_L)$}{
						$\delta'[v]-=\alpha$\;
					}
				
					\textbf{if $BD[s].d[v]==BD[s].d[w]$ then}
						$\alpha=0.0$\;
					\If{$BD[s].d[w]<BD[s].d[v]$}{
						$\alpha=\frac{BD[s].\sigma[w]}{BD[s].\sigma[v]}(1+BD[s].\delta[v])$\;
					}
					\textbf{if $(v,w) \neq (u_L,u_H)$ then}
						$EBC'[(v,w)]-= \alpha$\;
				}
				\If{$d'[v]==d'[w]$ \AndKw $BD[s].d[w]$$\neq$$BD[s].d[v]$}{
					Execute module in Alg.~\ref{algo:dep-mod-2}.
				}
			}
			\textbf{if $w\neq s$ then}
				$VBC[w]+=\delta'[w]-BD[s].\delta[w]$;
		}
		$level = level - 1$;
	}
	}
	\caption{Betweenness update for addition of an edge where $u_L$ rises one or more levels after addition.
	}\label{algo:L2-addition-neighbors}
\end{algorithm}

\begin{algorithm}[t]
\small{
	\If{$te[(v,w)]==N_T$}{
		$te[(v,w)]\leftarrow U_P;$
		$\alpha=0$\;
		\If{$BD[s].d[w]>BD[s].d[v]$}{
			$\alpha=\frac{BD[s].\sigma[v]}{BD[s].\sigma[w]}(1+BD[s].\delta[w])$\;
		}
		\If{$BD[s].d[w]<BD[s].d[v]$}{
			$\alpha=\frac{BD[s].\sigma[w]}{BD[s].\sigma[v]}(1+BD[s].\delta[v])$\;
		}
		$EBC[(v,w)]-= \alpha$\;
	}
}
\caption{EBC correction if endpoints were not at the same level before the change.}\label{algo:dep-mod-2}
\end{algorithm}

Algorithm~\ref{algo:L2-addition-neighbors} shows the details.
The algorithm begins by initializing the new distance of $u_L$ (line 1).
Similarly to the previous case, it adds all the new shortest paths to vertices found during its BFS traversal.
However, in this case there are structural changes in the \spdag due to the connection of two endpoints previously far apart.
The vertices reachable from $u_L$ might be pulled closer to the source $s$ along with $u_L$.
As a result, new shortest paths may emerge, old shortest paths may become obsolete, and distances from the source may change.
Therefore, the vertices do not inherit the shortest paths from their predecessors (line 3), rather, the shortest paths are computed during the modified BFS.

The structural changes that can happen in the \spdag are depicted in Figure~\ref{fig:cases-addition-deletion}.
Let us examine these cases for a vertex $x$ and its neighbor $y$.
Let a \emph{sibling} be a neighbor of vertex that is at the same distance from the source.
Before the addition, $x$ and $y$ could be either siblings (case 1, Fig.~\ref{fig:cases-addition-deletion}) or predecessor and successor (case 2).
If $y$ is now a predecessor of $x$ (case 1a), the algorithm adjusts the shortest paths of $x$ (line 5).
If $x$ was and still is a predecessor of $y$ (line 6), the new edge has caused both $x$ and $y$ to move closer to $s$ by the same amount (case 2a).
In this case, we update the distance from $s$ and insert $y$ in the BFS queue for further exploration (line 7).
If $y$ is now on the same level as $x$, but was not before the addition (case 2b), $y$ is added to the BFS for further exploration (line 10).
If $y$ moved two levels w.r.t. $x$ (case 2c), it will be discovered first in the BFS after the update (line 6).

Clearly there are no structural changes in the vertices at levels above $u_{L}$ (i.e., closer to the source).
The possible sub-cases examined cover all possible scenarios of how a pair of connected vertices (and thus their edge) can be found after the addition of the new edge.
The shortest paths and distances ($\sigma'[\cdot]$, $d'[\cdot]$) are updated in the way that the original Brandes' algorithm proposes.

In the dependency accumulation phase, the dependency score of all vertices examined is updated with the new number of shortest paths computed in the BFS phase.
This part of the algorithm is similar to the corresponding one in Algorithm~\ref{algo:L1-addition-removal-neighbors}.
However, there are important differences in the correction of the dependency for the edge betweenness centrality (lines 20--25, Alg.~\ref{algo:L2-addition-neighbors}).
Assuming $v$ is $x$ and $w$ is $y$, if both $x$ and $y$ remain at the same relative distance from the source, the dependency to be subtracted $\alpha$ is calculated in line 4 of Alg.~\ref{algo:dep-mod-1} (case 2a).
However, if $y$ moves closer (case 2c), then $y$ was a successor of $x$ but now it is a predecessor of $x$.
Therefore, we need to subtract the dependency on $y$.
The subtracted value is adjusted by switching $w$ with $v$ in the dependency accumulation formula (lines 21--22, Alg.~\ref{algo:L2-addition-neighbors}).

If the endpoints of the edge were at the same level before the addition (case 1) there is no need for correction since no dependency was accumulated on the edge (line 20, Alg.~\ref{algo:L2-addition-neighbors}).
If the endpoints are now at the same level but were not before (case 2b), the old dependency needs to be subtracted from the betweenness of the edge.
Also, the edge is marked not to be traversed again (Alg.~\ref{algo:dep-mod-2}).
In Alg.~\ref{algo:dep-mod-2}, if $w$ was a successor of $v$, the old dependency is calculated on line 4, whereas if $w$ was a predecessor of $v$, the old dependency is calculated on line 6.
The vertex betweenness centrality is updated on line 26 of Alg.~\ref{algo:L2-addition-neighbors} by adding the new dependency accumulated on the vertex $w$ and subtracting the old dependency.

In summary, all possible cases of structural changes in the \spdag below $u_{L}$ are covered by Alg.~\ref{algo:L2-addition-neighbors}, which correctly updates the betweenness scores and accompanying data structures of all affected vertices and edges.

\begin{figure}[t]
\begin{center}
	\includegraphics[scale=0.8]{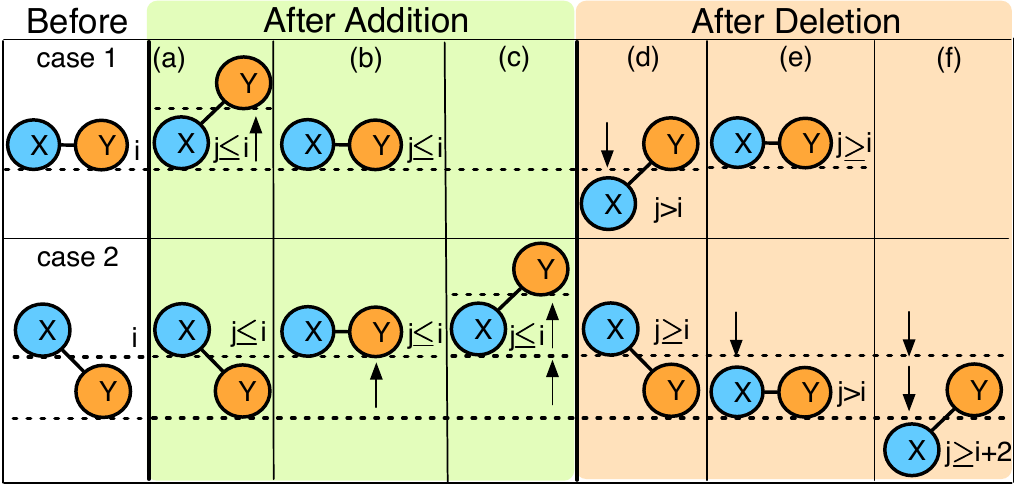}
\caption{Possible configurations of an edge before and after an update that causes structural changes.}
\label{fig:cases-addition-deletion}
\end{center}
\end{figure}

\subsection{Removal: 1 or more levels drop}\label{subsec:1lr}

\begin{algorithm}[t]
\small{
	$PQ[|V|] \leftarrow$ empty queues;
	$Q_{BFS} \leftarrow$ empty queue\;

	$first = V;$
	$Q_{BFS} \leftarrow u_L;$
	$t[u_L]\leftarrow N_P;$

	\While{$Q_{BFS}$ is not empty} {
		$v \leftarrow Q_{BFS}$\;
		\For{$ w \in neighbors(v)$} {
			\If{$BD[s].d[w]$+$1$$==$$BD[s].d[v]$ \AndKw $t[w]$==$N_T$ \AndKw $t[v]$$\neq$$P_V$} {
				$PQ[d[v]] \leftarrow v;$
				$t[v]\leftarrow P_V$ \tcp{a new pivot}

				\If{$first>BD[s].d[v]$}{
					$first=BD[s].d[v]$ \tcp{the first pivot}
				}
			}
			\ElseIf{$BD[s].d[w]$==$BD[s].d[v]$+$1$ \OrKw $BD[s].d[w]==BD[s].d[v]$}{
				\If{$t[w]==N_T$}{
					$Q_{BFS} \leftarrow w;$
					$t[w]=N_P$ \tcp{not a pivot}
				}
			}
		}
	}	
	\If{PQ[.] not empty}{
		Execute Algorithm~\ref{algo:L2+removal-neighbors-DepPivots} \tcp{ $\geq$ 1 pivots}
	}
	\ElseIf{PQ[.] is empty}{
		Execute Algorithm~\ref{algo:L2+removal-neighbors-disconn}	\tcp{Disconnected component}
	}
	}
	\caption{Pivot finder when an edge removal causes $u_L$ to drop one or more levels.}\label{algo:L2+removal-neighbors-pivots}
\end{algorithm}

Algorithm~\ref{algo:L2+removal-neighbors-pivots} shows how to handle an edge removal that causes $u_L$ to move one or more hops away from the source $s$.
First, the algorithm needs to find the pivots in the sub-dag rooted in $u_L$.
Then, distances from the source are updated by exploring the graph starting from the discovered pivots.
Finally, the usual dependency accumulation corrects the values of dependency and betweenness centrality.

The algorithm starts a BFS from $u_L$ to search for vertices that have predecessors still connected to the rest of the graph, i.e., outside the scope of the sub-dag rooted in $u_L$ (lines 6--9) and marks them as pivots ($P_V$) ($BFS_1$, orange, single-dotted triangle in Fig.~\ref{fig:examples-additions-deletions}b).
The rest are non-pivots ($N_P$).
Figure~\ref{fig:cases-addition-deletion} illustrates the possible cases that can be encountered by this BFS traversal in the case of edge removal.
Pivots are represented by vertex $y$ if $y$ remained in level $i$ (case 1d), or if $y$ remained in level $i+1$ (cases 2e and 2f).

After the first BFS finishes, Algorithm~\ref{algo:L2+removal-neighbors-DepPivots} starts a new BFS ($BFS_2$, blue, single-dashed triangle in Fig.~\ref{fig:examples-additions-deletions}b) from the pivots found.
This new BFS first corrects the distances of the vertices discovered (line 9), and their shortest path counts (line 12).
Furthermore, in lines 23-31 it adjusts the dependency and betweenness values in a similar fashion as described in Section~\ref{subsec:1li}, by covering all possible cases (Fig.~\ref{fig:cases-addition-deletion}) and following the $\delta$ arrows in Fig.~\ref{fig:examples-additions-deletions}b.

If $u_L$ has at least one sibling before the removal of the edge, we can use the same general technique presented above.
However, some optimizations to reduce the computation overhead are possible.
Indeed, the first search from $u_L$ to find the pivots can be avoided.
In fact, in this case all possible scenarios shown in Figure~\ref{fig:cases-addition-deletion} can be seamlessly found and resolved while adjusting the shortest paths (Algorithm~\ref{algo:L2+removal-neighbors-DepPivots}), since the starting pivots are the siblings of $u_L$, and the other pivots are all found during the BFS.
These optimizations are explained in detail next.

\begin{algorithm}[t]
\small{
		\BFS{ from first pivot point(s)}
		$Q_{BFS} \leftarrow$ empty queue;
		$Q_{BFS} \leftarrow PQ[first]$\;
		$next=first+1$\;

		\While{$Q_{BFS}$ not empty}{
			$v \leftarrow Q_{BFS};$
			$t[v]\leftarrow D_N;$
			$LQ[d'[v]] \leftarrow v;$
			$\sigma'[v]=0$\;

			\If{$next==d'[v]+1$}{
				$Q_{BFS} \leftarrow PQ[next];$
				$next = next + 1$
			}
			
			\For{$ w \in neighbors(v)$}{
				\If{$t[w]==N_P$}{
					$t[w]=D_N;$
					$d'[w]=d'[v]+1;$
					$Q_{BFS} \leftarrow w$\;
				}
				\textbf{else if $t[w]==P_V$ then}
					$t[w]\leftarrow D_N$\;
				\Else{
					\textbf{if $d'[w]+1==d'[v]$ then}
						$\sigma'[v]+=\sigma'[v]$\;
					\If{$d'[w]$==$d'[v]$ \AndKw $BD[s].d[w]$$\neq$$BD[s].d[v]$}{
						\If{$BD.d[w]>BD.d[v]$ \AndKw $t[w] \neq D_N$}{
							$t[w]\leftarrow D_N;$
							$LQ[d'[w]] \leftarrow w;$
							$Q_{BFS} \leftarrow w$\;
						}
					}
				}
			}
		}

		\Dep{}
		$\delta'[u_H]=BD[s].\delta[u_H]-\frac{\sigma[u_H]}{\sigma[u_L]}(1+BD[s].\delta[u_L])$\;
		$LQ[d'[u_H]] \leftarrow u_H;$
		$te[e]$$\leftarrow$$N_T, e \in E;$
		$t[u_H]$$\leftarrow$$U_P;$
		$level=|V'|$\;
		\While{$level>0$}{
			\While{$LQ[level]$ not empty}{
				$w \leftarrow LQ[level]$\;
				\For{$ v \in neighbors(w)$}{
					\If{$d'[v] < d'[w]$}{
					
						Execute module in Alg.~\ref{algo:dep-mod-1};
						$\alpha=0;$\;
						\If{$BD[s].d[w]>BD[s].d[v]$}{
							$\alpha=\frac{BD[s].\sigma[v]}{BD[s].\sigma[w]}(1+BD[s].\delta[w])$\;
						}
						\ElseIf{$BD[s].d[w]<BD[s].d[v]$}{
							$\alpha=\frac{BD[s].\sigma[w]}{BD[s].\sigma[v]}(1+BD[s].\delta[v])$\;
						}
						\textbf{if $t[v]==U_P$ then}
							$\delta'[v]-=\alpha$\;
						$EBC'[(v,w)]-=\alpha$\;
					}
					\If{$d'[v]$==$d'[w]$ \AndKw $BD[s].d[w]$$\neq$$BD[s].d[v]$}{
						Execute module in Alg.~\ref{algo:dep-mod-2}.
					}
				}
				\textbf{if $w\neq s$ then}
					$VBC[w]+=\delta'[w]-BD[s].\delta[w]$\;
			}
			$level = level - 1$
		}
	}
	\caption{Betweenness update for edge removal where $u_L$ drops one or more levels after removal.}\label{algo:L2+removal-neighbors-DepPivots}
\end{algorithm}

\subsection{Removal: 1 level drop (opimization)}\label{subsec:1slr}

Algorithms~\ref{algo:L1-removal-neighbors-p1} and~\ref{algo:L1-removal-neighbors-p2} show in detail the steps needed to be performed by the framework when an edge removed forces $u_L$ to drop only one level from the source, as an optimization from the previous method that is generic and covers one or more levels drop.

In this case, $u_L$ does not have any other predecessors in the previous level.
Note that this one level drop of $u_L$ can lead to subsequent changes to many of $u_L$'s successors with respect to their distance from source $s$.
By executing Algorithm~\ref{algo:L1-removal-neighbors-p1}, a BFS starts from $u_L$ which targets on first fixing the distances of the found vertices, and after that, adjusting their shortest path counts.

During the BFS traversal from $u_L$, various sub-cases are encountered depending on the neighbors of each vertex and where it was positioned with respect to the vertex under examination (lines 8, 11 and 20).
Figure~\ref{fig:cases-addition-deletion} illustrates these possible sub-cases.
Note that in some cases (1e and 2d), there is no relative change: the vertices remain in the same distance difference as they were before the removal (regardless if they stayed in the same position or moved together downwards).
However, there are also cases where one of the two vertices moves and the other does not, due to their predecessors (Figure~\ref{fig:cases-addition-deletion}, cases 1d and 2e).
In such cases, not only the distance of the moved vertex must be fixed, but also the shortest paths counts from source.
All these sub-cases are examined in Algorithm~\ref{algo:L1-removal-neighbors-p1}.

Observe that we start the BFS from $u_L$ to minimize the scope of changes, i.e., to investigate only the sub-dag directly under $u_L$.
Therefore, a simple BFS is not enough to perform both of these adjustments (i.e., distance and shortest paths) in one pass as before.
Instead, in some cases, we need to check the neighbors $z$ of the neighbor $w$ under examination, if this neighbor is not properly adjusted yet (lines 25--27).
If we had executed a BFS from the same-level neighbors of $u_L$, we could potentially perform a one-pass BFS.
However, this would lead to greater costs with respect to how many vertices are unnecessarily touched on the way down with the BFS and then on the way up during the accumulation phase.

Any vertices discovered that did not move but have neighbors who moved, are pivoting points and are placed into the same-level BFS queue to be examined in this BFS level (line 31).
Otherwise, they are placed in the next-level BFS queue (lines 29 or 33).
If they have moved, their distance is also corrected (line 33).
Finally, Algorithm~\ref{algo:L1-removal-neighbors-p2} is executed to adjust the dependencies of the affected vertices and edges.

\begin{algorithm}[t]

	\BFSul{}
	$te[e]=N_T, e \in E$\;
	$Q_{Same} \leftarrow$ empty queue; $Q_{Next} \leftarrow$ empty queue\;
	$d'[u_L]++;$ $t[u_L]=D_N;$ $Q_{Same} \leftarrow u_L$\;
	$LQ[d[u_L]] \leftarrow u_L;$
	$d'[u_L]=BD[s].d[u_H]+1$\;
		
	\While{$Q_{Same}$ not empty}{
		$v \leftarrow Q_{Same};$	$\sigma'[v]=0;$	$LQ[d'[v]] \leftarrow v$\;
		
		\For{$w \in neighbors(v$)}{
			\If{$BD[s].d[w]==BD[s].d[v]$}{
				\If{$t[w]\neq D_N$ \AndKw $t[v]==D_N$}{
					$\sigma'[v]+=\sigma'[w]$\;
				}
			}
			\If{$BD[s].d[w]+1==BD[s].d[v]$}{
				\If{$t[w]\neq D_N$ \AndKw $t[v]==P$}{
					$\sigma'[v]+=\sigma'[w]$\;
				}
				\If{$t[w]==D_N$ \AndKw $t[v]==D_N$}{
					$\sigma'[v]+=\sigma'[w]$\;
				}
				\If{$t[w]==D_N$ \AndKw $t[v]==P$}{
					\If{$te[(v,w)]==N_T$}{
						$\alpha=\frac{ND[s].\sigma[w]}{ND[s].\sigma[v]}(1+ND[s].\delta[v])$\;
						$te[(v,w)]=D_N;$
						$EBC[(v,w)]-= \alpha$\;
					}
				}

			}
			\If{$BD[s].d[w]==BD[s].d[v]+1$ \AndKw $t[w]==N_T$}{
					\If{$t[v]==P$} {
                            			$t[w]=P$\;
					}
					\Else {
						$t[w]=D_N$\;
						\For{$z \in neighbors(w$)}{
							\If{$BD[s].d[z]+1==BD[s].d[w]$ $\AndKw$ $t[z]\neq D_N$}{
								$t[w]=P;$	$break$\;
							}
						}
					}
					\If{$t[w]==P$ $\AndKw$ $t[v]==P$}{
						$Q_{N} \leftarrow w;$
					}
					\ElseIf{$t[w]==P$ $\AndKw$ $t[v]==D_N$} {
						$Q_{S} \leftarrow w;$
					}
					\ElseIf{$t[w]==D_N$} {
						$Q_{N} \leftarrow w$;
						$d'[w] = d'[w] + 1$\;
					}
			}
		}
		\If{$Q_{S}$ is empty $\AndKw$ $Q_{N}$ not empty} {
			$Q_{S} \leftarrow Q_{N};$	$Q_{N} \leftarrow$ empty queue\;
		}
	}
	\caption{Betweenness update for removal of an edge where $u_L$ drops one level (BFS part).}\label{algo:L1-removal-neighbors-p1}
\end{algorithm}

\begin{algorithm}[t]
	\Dep{}
	$\delta'[u_H]=BD[s].\delta[u_H]-\frac{\sigma[u_H]}{\sigma[u_L]}(1+BD[s].\delta[u_L])$\;
	$LQ[d'[u_H]] \leftarrow u_H;$	$t[u_H]=U_P;$
	$level=|V'|$\;
	\While{$level>0$}{
		\While{$LQ[level]$ not empty}{
			$w \leftarrow LQ[level]$\;
			\For{$v \in neighbors(w$)}{
				\If{$d'[v] < d'[w]$}{
					$c=\frac{\sigma'[v]}{\sigma'[w]}(1+\delta'[w]);$	$\delta'[v]+=c;$
					$\alpha=0$\;
					\If{$BD[s].d[v]<BD[s].d[w]$}{
						$\alpha=\frac{BD[s].\sigma[v]}{BD[s].\sigma[w]}(1+BD[s].\delta[w])$\;
					}
					\If{$t[v]==0$}{
						$t[v]=1;$	$LQ[level-1] \leftarrow v$\;
						$\delta'[v]+=BD[s].\delta[v]$\;
					}
					$EBC'[(v,w)]+=c-\alpha$\;

					\If{$t[v]==1$}{
						$\delta'[v]-=\alpha$\;
					}
				}
			}
			\If{$w\neq s$}{
				$VBC[w]+=\delta'[w]-BD[s].\delta[w]$\;
			}
		}
		$level = level - 1$
	}

	\caption{Betweenness update for removal of an edge where $u_L$ drops one level (depend. acc. part).}\label{algo:L1-removal-neighbors-p2}
\end{algorithm}

\subsection{Removal: Disconnected Component}\label{subsec:discon}

If no pivoting points are found from Algorithm~\ref{algo:L2+removal-neighbors-pivots}, then it means the sub-dag under $u_L$ is a disconnected component and thus unreachable from source $s$.
Therefore, Algorithm~\ref{algo:L2+removal-neighbors-disconn} is executed to re-initialize the data structures and correct the betweenness values of vertices and edges.
First, a BFS is started from $u_L$ and to initialize the data structures of the vertices found as well as to adjust the betweenness centrality values for vertices and edges.
Second, the algorithm backtracks the $LQ$ queues from $u_H$ up and adjusts all dependency values and betweenness values of vertices and edges in the other disconnected component.
If indeed the removal disconnects this sub-graph (either a portion of it or just one vertex into a singleton), this part of the framework will be executed for every source of the graph, and the data structures and betweenness scores will be adjusted accordingly in the storage files.
We note that this step is needed for every source and cannot be avoided, e.g., by treating it as a special case at the beginning of the framework.

\begin{algorithm}[t]

		\BFSul{}
		$Q_{BFS} \leftarrow$ empty queue\;
		$Q_{BFS} \leftarrow u_L;$		$d'[u_L]=-1;$	$\sigma'[u_L]=0;$	$\delta'[u_L]=0$\;
		
		\While{$Q_{BFS}$ is not empty} {
			$v \leftarrow Q_{BFS}$\;
			\For{$ w \in neighbors(v)$}{
				\If{$BD[s].d[w]==BD[s].d[v]+1$}{
					\If{$t[w]==D_N$}{
						$Q_{BFS} \leftarrow w;$	$t[w]=M;$	$d'[w]=-1;$	$\sigma'[w]=0;$	$\delta'[w]=0$\;
					}
					$\alpha=\frac{BD[s].\sigma[v]}{BD[s].\sigma[w]}(1+BD[s].\delta[w])$\;
					$VBC[w]-=\alpha;$	$EBC'[(v,w)]-=\alpha$\;
				}
			}
		}		
		\Dep{}
		$\delta'[u_H]=BD[s].\delta[u_H]-\frac{\sigma[u_H]}{\sigma[u_L]}(1+BD[s].\delta[u_L])$\;
		$LQ[d'[u_H]] \leftarrow u_H;$
		$t[u_H]$$\leftarrow$$U_P;$
		$level=|V'|$\;
		\While{$level>0$}{
			\While{$LQ[level]$ not empty}{
				$w \leftarrow LQ[level]$\;
				\For{$v \in neighbors(w)$}{
					\If{$d'[v] < d'[w]$}{
						Execute module in Alg.~\ref{algo:dep-mod-1}.\;
						\textbf{if $t[v]==U_P$ then}
							$\delta'[v]-=\alpha$\;
						$EBC'[(v,w)]-=\alpha$\;
					}
				}
				\If{$w\neq s$}{
					$VBC[w]+=\delta'[w]-BD[s].\delta[w]$\;
				}
			}
			$level = level - 1$
		}
	
	\caption{Betweenness update for removal of an edge where $u_L$'s subdag forms a disconnects component.}\label{algo:L2+removal-neighbors-disconn}
\end{algorithm}

\section{Scalability}\label{sec:scalable}

The algorithm described previously is able to update the betweenness centrality of a graph incrementally.
Nevertheless, to attain our goal of a useful and practical framework for real-world deployment, the algorithm is not enough.
Indeed, a number of scalability issues needs to be addressed in order to have a practical tool.

The space complexity of the algorithm is quadratic in the number of vertices.
When dealing with large graphs, the space requirements can easily outgrow the available main memory.
In this case, we can use the disk to store the required data structures, as our algorithm allows for an efficient out-of-core implementation.

Despite this feature, the space requirements for very large graphs can still outgrow the disk.
Furthermore, disk access will become a bottleneck due to reading and writing large amounts of data.
A simple solution is to divide the execution across multiple machines with multiple disks.
This solution not only allows the framework to scale to larger inputs, but also leads to improved speedup over the sequential version.

\subsection{Out-of-core organization}\label{subsec:out-of-core}

The removal of the predecessors lists from the algorithm reduces its space complexity by $\mathcal{O}(nm)$.
As an additional benefit, it leaves the algorithm with no variable-length data structure.
Indeed, our algorithm stores only three fixed size data structures per vertex: the distance from the source $d[\cdot]$ (1 byte), the number of shortest paths from the source $\sigma[\cdot]$ (2 bytes) and the accumulated dependency $\delta[\cdot]$ (8 bytes).
This organization allows an efficient data layout on disk.
The graph is loaded and kept in memory, allowing fast random access of vertex neighborhoods.

We encode $BD[\cdot]$ in binary format on disk.
For each source $s$, we store the data for each other vertex in a columnar fashion, i.e., we store on disk all the distances, then all the numbers of shortest paths, and finally the dependency values, in order: $\left\{BD[s].d[\cdot], BD[s].\sigma[\cdot], BD[s].\delta[\cdot]\right\}$.
We avoid storing the vertex IDs for sources and destinations by storing the data structures sequentially on disk, and inferring the ID from the order.
Overall, for each source the algorithm stores $3$ arrays, with sizes $1 \times n$, $2 \times n$ and $8 \times n$ bytes, respectively.
Because of the binary format, each array can be read by using file channels and byte buffers, loaded directly in memory and be ready for use.
This optimization avoids memory allocations and memory copies during the execution of the algorithm.
In practice, the computation happens at the speed of sequential disk access.
In the future we plan to explore compression schemes to reduce the space and disk access overhead.

As explained in the previous sections, the work done by the algorithm depends on $dd(u_L, u_H)$.
Therefore, after loading the distances from disk, we check the distance for the endpoints $u_H$ and $u_L$.
If they are at the same distance ($dd=0$), we skip directly to the next source without loading the rest of the data structure.
This operation is efficient because the data structures have fixed size so the offset to skip to reach the beginning of the next source is constant.
Otherwise, the arrays are loaded in memory and Algorithm~\ref{algo:iterative-overview-addition-deletion-neighbors} is executed.
When a source is covered, the algorithm writes the arrays back to disk, in place and sequentially.

\subsection{Parallelization}

The out-of-core version presented in the previous section is slower than the in-memory one, but enables the framework to scale to large graphs.
However, the space requirements for real graphs can be staggering (for a graph of 1M vertices we need $\approx11$TB of space).
Although nowadays is fairly easy to have such an amount of disk storage available on servers, reading and writing this amount of data would incur a significant overhead.

To solve this issue, we take advantage of the parallel nature of our framework and propose to distribute the computation on a cluster of shared-nothing machines.
Let $p$ be the number of available machines in the system.
We distribute the data structure $BD[\cdot]$ evenly among the $p$ machines, i.e., each machine will be allocated $\approx n/p$ sources.
This parallelization is possible since each source can be examined and updated independently, and the partial betweenness scores can be summed at the end.

The distributed version of the framework has multiple advantages.
First, the space requirements per machine are reduced to $\mathcal{O}(n^2/p)$.
For example, if $100$ machines are available, for a graph with $1$M vertices the storage needed is just 150GB per machine, which is perfectly reasonable by today's standards.
Second, the overall work of the algorithm is divided among $p$ processors, leading to a theoretical $p$ speedup over single machine execution.
Furthermore, the disk access workload is distributed in a balanced fashion across multiple disks.
As a result, the disk access speed of the framework is $p$ times faster.
As we show in the next section, the speedup of the framework is indeed almost ideal.

\subsection{Online betweenness updates}

Real graphs are constantly evolving with addition of new vertices and edges, and removal of existing ones.
Updating the betweenness centrality in real-time is extremely challenging, given its computational cost.
However, due to the inherent parallelism of our design, the framework can scale not only to large graphs but also to rapidly changing ones.

Each of the $p$ available machines is responsible for updating the data structures and partial centrality scores for $n/p$ sources.
The system can monitor the average time $\overline{t_{S}}$ needed for each machine to process a source, given the addition of a new edge or the removal of an existing one.
Thus, the average time $\overline{t_{U}}$ to produce updated betweenness scores for all vertices and edges upon arrival of an update is
$\overline{t_{U}} ={\overline{t_S} \times n / p} + \overline{t_{M}}$, where $\overline{t_{M}}$ is the average time needed to merge the results.

Assuming that an evolving graph has an average rate of updates per time unit $F = 1 / t_{I}$, the system can always (on average) produce updated betweenness scores before the new time period, if $\overline{t_{U}} < t_{I}$.
However, if the system measures an increased rate of arrival $F' > F$, it can adjust the number of machines to $p' > {\overline{t_S} \times n / (t'_{I}-\overline{t_{M}})}$ to guarantee online updates, assuming that the average time per source remains unchanged, and that $(t'_{I} > \overline{t_{S}} + \overline{t_{M}})$, i.e., the inter-arrival time is larger than the inherent serial part of the algorithm.

\subsection{A MapReduce embodiment}

Various paradigms can be used to embody the proposed parallel framework, such as parallel stream processing engines (e.g., Storm, S4, and Samza).
Due to its ease of use and popularity, we deploy and experiment with a MapReduce embodiment on a Hadoop cluster.

\begin{figure}[t]
\begin{center}
	\includegraphics[scale=0.38]{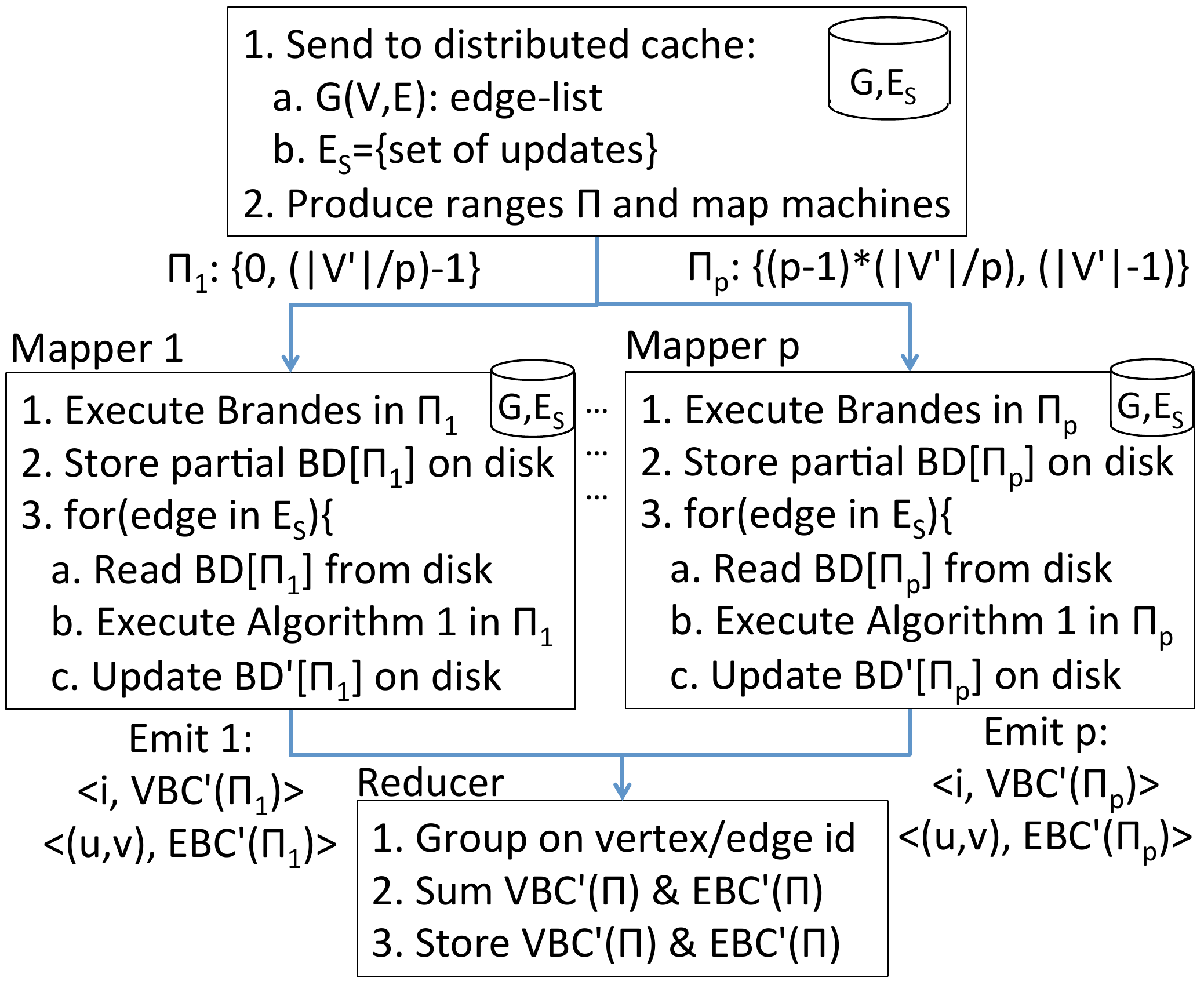}
\caption{ MapReduce version of our framework.}
\label{fig:mr-framework}
\end{center}
\end{figure}

Figure~\ref{fig:mr-framework} illustrates a MapReduce adaptation of our algorithmic framework.
The graph $G(V,E)$ and the set of updates $E_S$ (new edges to be added and existing edges to be removed) are replicated on all machines via distributed cache, and loaded in memory.
We generate an input for each mapper $i$ that represents a partition $\Pi_i$ of the graph.
Each partition is comprised of two integers that represent the first and last ID of the range of sources for which the particular mapper $i$ is responsible.
The data structures $BD[\Pi_i]$ created during step 1 are stored locally on the disk of each machine.

The Map function processes all edges in $E_S$ in sequence and updates the betweenness centrality.
For each update, it emits key-value pairs of vertex or edge IDs together with their partial betweenness centrality ($PBC$) by source $s$, i.e., $\langle id, \textsc{vbc}_{s}(id)|\textsc{ebc}_{s}(id) \rangle$, where $id$ is either a vertex or an edge identifier.
All the intermediate pairs are sent to the reducers who are responsible for producing the final aggregated betweenness results for all vertices and edges.
Each Reduce function aggregates the partial betweenness score of one element (vertex or edge) in the graph.
The final value of the computation is the new betweenness score for each element of the graph after the set of updates $E_S$ is applied.

The signatures of the functions are as follows:
\begin{eqnarray*}
\textsf{Map}: & &   \langle \Pi_i, \emptyset \rangle \  \rightarrow \ \left[ \langle id, \textsc{pbc}_{s}(id) \rangle \  \forall id \in G, \forall s \in \Pi_i \right] \\
\textsf{Reduce}: & & \langle id, \left[ \textsc{pbc}_{s}(id), \ldots \right] \ \forall s \in V \rangle  \rightarrow \langle id, BC(id) \rangle
\end{eqnarray*}

\section{Experimental Results}\label{sec:exp-methodology-results}

We evaluate our algorithmic framework on real and synthetic graphs to assess its performance over Brandes' algorithm and its ability to scale to large graphs.

\spara{Datasets.}
We use two types of graphs: synthetic and real.
The synthetic graphs are created with a synthetic social graph generator~\citep{sala10graphmodels}, which produces graphs with properties, such as degree distribution and clustering coefficient, similar to real social graphs.
The synthetic graphs enable us to experiment with graphs that maintain properties of real social graphs while being able to freely increase their size (see Table~\ref{tab:all-nets}).

Table~\ref{tab:all-nets} also reports the details of the real graphs we use.
They are taken from the KONECT collection\footnote{\url{http://konect.uni-koblenz.de/graphs}} and come from different domains:
\dataset{wiki-elections} (\textsf{WE} for short, election votes for Wikipedia admins),
\dataset{epinions} (\textsf{EP}, trust among Epinion users),
\dataset{facebook} (\textsf{FB}, friendships among Facebook users),
\dataset{slashdot} (\textsf{SD}, replies among Slashdot users),
\dataset{dblp} (co-authorships among scholars),
and \dataset{amazon} (\textsf{AMZ}, product ratings by Amazon users).
To make the results comparable between real and synthetic graphs as well as with previous works, we use the largest connected component (LCC) of the real graphs.

\begin{table}[t!]
\caption{Description of the graphs used. AD: average degree, CC: clustering coefficient, ED: effective diameter.}
\centering
\small
\tabcolsep=0.12cm
\begin{tabular}{lrrrrrr}
\toprule
\hspace{-4mm}& Dataset	&	$|V|(LCC)$	&	$|E|(LCC)$	&	AD		&	CC		&	ED		\\
\midrule
\multirow{4}{-2mm}{\begin{sideways} \textsf{\textcolor[rgb]{0.50,0.51,0.53}{synthetic}} \end{sideways}}
\hspace{-4mm}& 1k		&	\num{1000}	&	\num{5895}	&	11.8		&	0.263	&	5.47		\\
\hspace{-4mm}& 10k		&	\num{10000}	&	\num{58539}	&	11.7		&	0.219	&	6.56		\\
\hspace{-4mm}& 100k		&	\num{100000}	&	\num{587970}	&	11.8		&	0.207	&	7.07		\\
\hspace{-4mm}& 1000k	&	\num{1000000}	&	\num{5896878}	&	11.8		&	0.204	&	7.76		\\
\midrule
\multirow{6}{-2mm}{\begin{sideways} \textsf{\textcolor[rgb]{0.50,0.51,0.53}{real-world}} $\;$ \end{sideways}}
\hspace{-4mm}& \dataset{wikielections}		&	\num{7066}	&	\num{100780}	&	8.3		&	0.126	&	3.78		\\
\hspace{-4mm}& \dataset{slashdot}			&	\num{51082}	&	\num{117377}	&	51.1		&	0.006	&	5.23		\\
\hspace{-4mm}& \dataset{facebook}			&	\num{63392}	&	\num{816885}	&	63.7		&	0.148	&	5.62		\\
\hspace{-4mm}& \dataset{epinions}			&	\num{119130}	&	\num{704571}	&	12.8		&	0.081	&	5.49		\\
\hspace{-4mm}& \dataset{dblp}				&	\num{1105171}	&	\num{4835099}	&	8.7		&	0.6483	&	8.18		\\
\hspace{-4mm}& \dataset{amazon}			&	\num{2146057}	&	\num{5743145}	&	3.5		&	0.0004	&	7.46		\\
\bottomrule
\end{tabular}
\label{tab:all-nets}
\end{table}

\spara{Graph updates.}
For edge addition in the synthetic graphs, we generate the stream of added edges $E_S$ by connecting $100$ random unconnected pairs of vertices.
For the real graphs, each edge has an associated timestamp of its real arrival time, so we simply replay them in order.
For edge removal in the synthetic graphs, we randomly select $100$ existing edges to construct the stream of removed edges $E_S$.
For real graphs, we remove the last $100$ edges that are added in each graph and do not create a graph partition.
The use of real arrival times in graphs is an important difference from previous studies of betweenness centrality updates.
This scenario allows us to simulate the evolution of a real system, and thus assess the capability of our framework to update the betweenness centrality \emph{online}.

\spara{Implementation.}
We implement our algorithmic framework in Java and use the JUNG graph library\footnote{\url{http://jung.sourceforge.net}} for basic graph operations and maintenance.
For the out-of-core version, we store $BD[\cdot]$ in a single file, and read it sequentially in memory source by source.
If any update is needed for the current source, it is performed in place on disk rather than overwriting the whole file.
This enhancement limits the writes on disk to a minimum.
In the experiments, we compare the performance of three versions of the framework: ($1$)~in memory with predecessors lists (\textsf{MP}), ($2$)~in memory without predecessors lists (\textsf{MO}), ($3$)~on disk without predecessors lists (\textsf{DO}).

\textbf{Infrastructure.}
For the single-machine version (both in-memory and out-of-core), we use high-end servers with 8-core Intel Xeon $@2.4$GHz CPU and $50$GB of RAM each.
For the parallel version, we use a Hadoop cluster with hundreds of machines with 8-core Intel Xeon $@2.4$GHz CPU and $24$GB of RAM each.
We report the average performance over $10$ executions of the algorithms for each experimental setup.

\subsection{Speedup over Brandes' algorithm}

\spara{Predecessors list optimization.}
Figure~\ref{fig:rm_speedups_additions_cdf} presents the cumulative distribution function (CDF) of the speedup over Brandes' algorithm when testing the three versions of the framework on edge addition.
Each point in the graph represents the speedup when adding one of the $100$ edges in the stream (averaged over $10$ runs).
The results show that removing the predecessors lists can actually boost the performance (the \textsf{MO} version is always faster than the \textsf{MP} version).
As the algorithm does not need to create the lists nor to maintain them, the overhead, and thus the overall execution time, is reduced.

\begin{figure*}[t]
\begin{center}
	\includegraphics[scale=1]{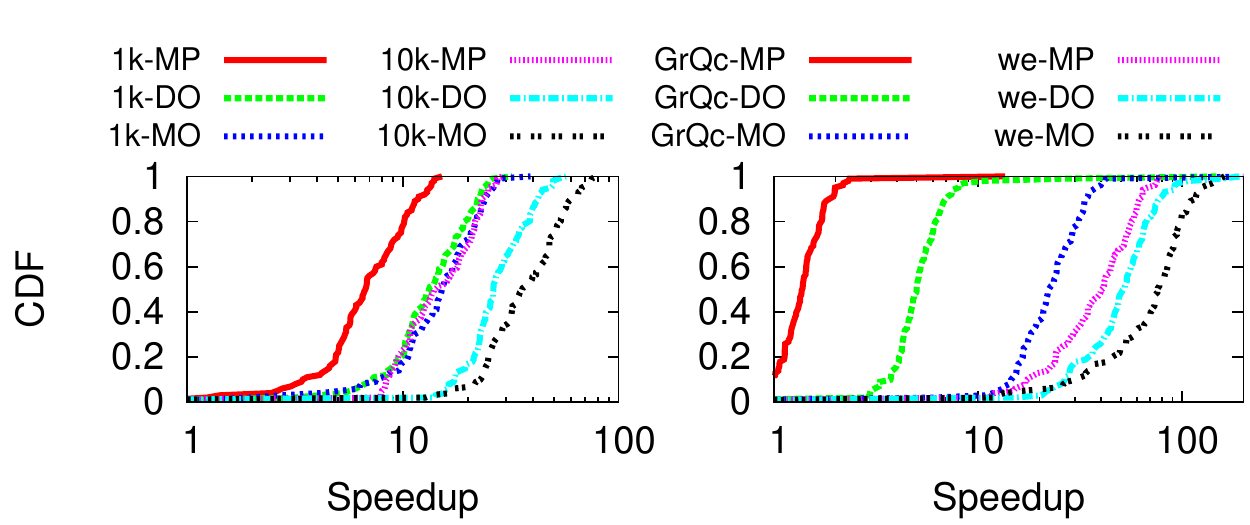}
\caption{Speedup of the framework's 3 versions on synthetic and real graphs executed on single machines (addition).}
\label{fig:rm_speedups_additions_cdf}
\end{center}
\end{figure*}

\spara{Related work comparison.}
As shown in Table~\ref{tab:summary-exp-results-rel-work}, the average performance of our framework is comparable to the reported results of previously proposed techniques.
The method in~\cite{kas13betweenness} shows faster results in networks with low clustering coefficient where changes do not affect many vertices.
However, in social networks with high clustering, the method by~\citet{kas13betweenness} and especially QUBE~\cite{lee12qube} are slow due to high cost in updating vertex centrality (in QUBE, many vertices might be included in the minimum union cycle to be updated).

We also compared with the method by~\citet{green12iter-betweenness} using a version of their code.
By taking into account implementation differences (Java vs. C), the speedups observed are comparable to our method.
However, under limited main memory, the method by~\citeauthor{green12iter-betweenness} fails to tackle a medium-sized network like \dataset{slashdot}. 
Instead, as demonstrated later in Table~\ref{tab:summary-exp-results}, our framework can handle even larger datasets using out-of-core techniques with small main memory footprint and significant speedups.

Additionally, these past techniques compute only vertex centrality, while our method computes both vertex and edge centrality with the shown speedups.

\begin{table}[t]
\caption{Speedup comparison with related work.}
\tabcolsep=0.12cm
\centering
\small
\begin{tabular}{r|c|c|rrr}
\toprule
Dataset				&	$|V|$	&	MO avg (max)	&	 \cite{kas13betweenness}	&	\cite{lee12qube} 	&	\cite{green12iter-betweenness} \\
\midrule
\dataset{wikivote}		&	7k	&	75 (181)	&		&	3	&			\\
\dataset{contact}		&	10k	&	75 (153)	&		&	4	&			\\
\dataset{UCI (fb-like)}	&	2k	&	32 (90)	&	18	&		&			\\
\dataset{ca-GrQc}		&	4k	&	31 (378)	&	68	&	2	&	40		\\
\dataset{ca-HepTh}		&	8k	&	42 (80)	&	358	&		&	40		\\
\dataset{adjnoun}		&	.1k	&	48 (172)	&		&		&	20		\\
\dataset{ca-CondMat}	&	19k	&	94 (395)	&		&		&	109		\\
\dataset{as-22july06}		&	23k	&	70 (291)	&		&		&	61		\\
\dataset{slashdot} (50GB)	&	51k	&	88 (178)	&		&		&	X		\\
\bottomrule
\end{tabular}
\label{tab:summary-exp-results-rel-work}
\end{table}

\spara{Out-of-core performance.}
When $BD[\cdot]$ is stored on disk (\textsf{DO}) rather than in memory (\textsf{MO}), we observe a decrease of the speedup due to the slower access time of the disk.
Overall, the \textsf{DO} version is more than $10 \times$ faster than Brandes' for the $1$k and more than $30 \times$ for the $10$k graph (median values).
The time to process a single edge depends heavily on which parts of the graph it connects and how many structural changes it produces.
The in-memory version is CPU bound, so this variability is reflected in the execution time.
On the other hand, the out-of-core version is I/O bound, and the execution time is dominated by disk access and the variability of the CPU time spent becomes latent.
In the remainder of this section we use the \textsf{DO} version.
Key speedup results are summarized in Table~\ref{tab:summary-exp-results}.

\begin{table}[b]
\caption{Summary of key speedup results.}
\centering
\small
\begin{tabular}{r|rrr|rrr}
\toprule
Dataset				&	\multicolumn{3}{c|}{Addition}	&	\multicolumn{3}{c}{Removal}	\\
					&	Min & Med & Max			&	Min & Med & Max		\\
\midrule
1k					&	3	&	12	&	23		&	2	&	10	&	19	\\
10k					&	16	&	34	&	62		&	2	&	35	&	155	\\
100k					&	21	&	49	&	96		&	4	&	45	&	134	\\
1000k				&	5	&	10	&	20		&	1	&	12	&	78	\\
\midrule
\dataset{wikielections}	&	9	&	47	&	95		&	1	&	45	&	92	\\
\dataset{slashdot}		&	15	&	25	&	121		&	8	&	24	&	127	\\
\dataset{facebook}		&	10	&	66	&	462		&	1	&	102	&	243	\\
\dataset{epinions}		&	24	&	56	&	138		&	2	&	45	&	90	\\
\dataset{dblp}			&	3	&	8	&	15		&	3	&	8	&	429	\\
\dataset{amazon}		&	2	&	4	&	15		&	2	&	3	&	5	\\
\bottomrule
\end{tabular}
\label{tab:summary-exp-results}
\end{table}

\begin{figure*}[t]
\begin{center}
	\includegraphics[scale=1]{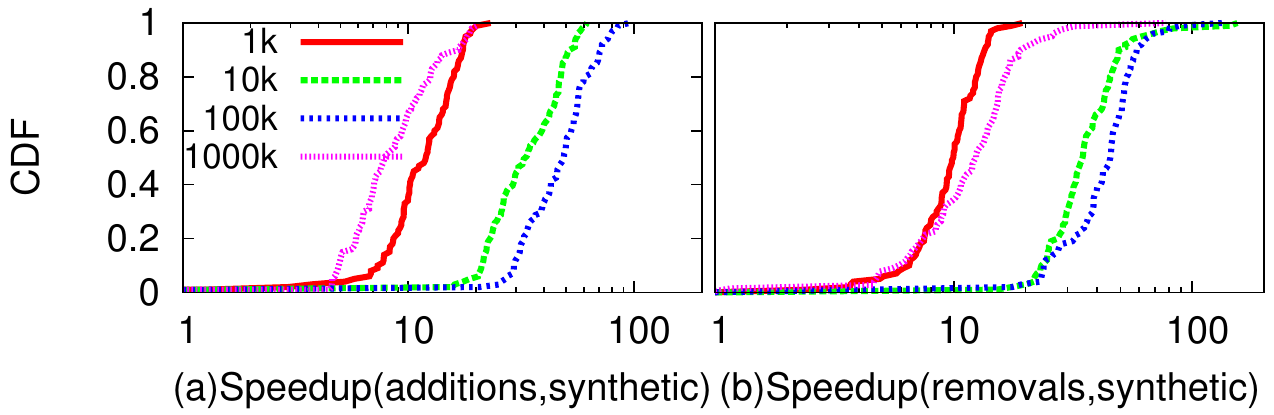}
	\includegraphics[scale=1]{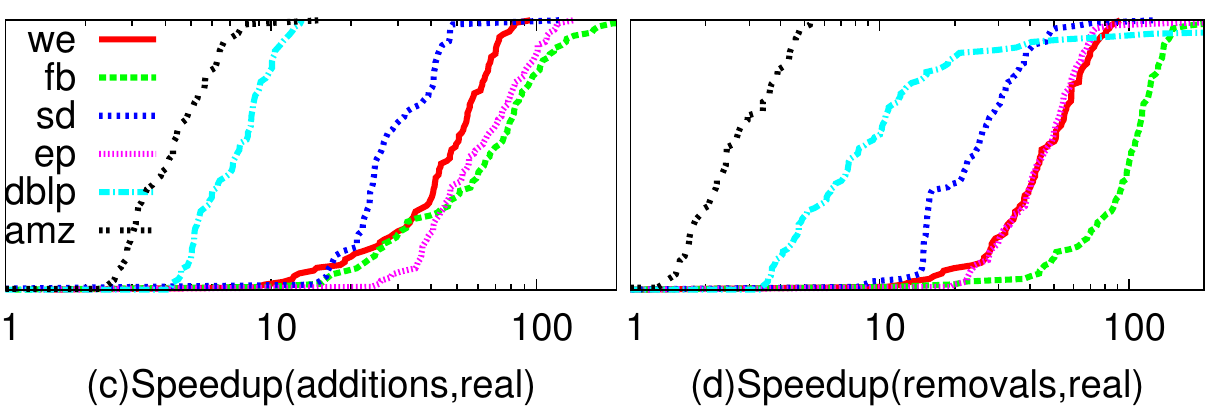}
\caption{Speedup of \textsf{DO} version on synthetic/real graphs executed on a MapReduce cluster (additions/removals).}
\label{fig:mr_speedups_synthetic_real_additions_removals_cdf}
\end{center}
\end{figure*}

\spara{MapReduce speedup.}
Figure~\ref{fig:mr_speedups_synthetic_real_additions_removals_cdf}(a) shows the CDF of speedup over Brandes' algorithm when executing the \textsf{DO} version on a MapReduce cluster for addition of edges.
In the experiment, we adjust the number of mappers so that each mapper is assigned $1$k sources per graph.
Brandes' algorithm is compared with the cumulative execution time of our algorithm, i.e., the sum of execution times across all mappers and reducers.
By increasing the graph size from $1$k to $100$k vertices, the median speedup increases from $\approx 10$ to $\approx 50$.
When increasing the graph size to $1000$k vertices, the median speedup drops to $\approx 10$.
Compared to the experiments on a single machine, there is an increase in the variability of the framework's performance when accessing the disks on the MapReduce cluster.
This effect is partly due to contention on disk access on the cluster caused by concurrent jobs, as well as increased computation load per machine and source.
Overall, the use of parallel execution leads to improved speedups for larger graphs that would be impossible to process on a single machine.
As an additional benefit, we also get reduced wall-clock time.

Figure~\ref{fig:mr_speedups_synthetic_real_additions_removals_cdf}(b) shows the CDF of speedup over Brandes' algorithm on a MapReduce cluster for removal of edges.
The setup is similar to the previous experiment.
By increasing the graph size from $1$k to $100$k vertices, the median speedup increases from $\approx 10$ to $\approx 45$.
When increasing the graph size to $1000$k vertices, the median speedup drops to $\approx 12$.
In this case, the speedup is slightly higher than when adding edges, because the removal of edges reduces the shortest paths between vertices and causes slightly less computational load.

\spara{Real graph structure.}
Figures~\ref{fig:mr_speedups_synthetic_real_additions_removals_cdf}(c) and (d) show the CDF of speedup over Brandes' algorithm for the real graphs when adding or removing edges, respectively.
Also in this case we adjust the number of mappers so that each mapper is assigned $1$k sources per graph.

In the edge addition, \dataset{facebook} exhibits the highest variability with a median speedup of $\approx 66$.
In the edge removal, \dataset{dblp} exhibits higher variability than \dataset{facebook} with a median speedup of $\approx 8$.
When adding edges on \dataset{slashdot}, which has a number of edges similar to \dataset{wikielections} but fewer vertices, our framework exhibits lower variability and smaller maximum speedup than on \dataset{wikielections}.
It also performs better on \dataset{facebook} than \dataset{slashdot}, both in addition and removal of edges, even if the two graphs have approximately the same number of vertices.
One reason may lie in the higher clustering coefficient of \dataset{wikielections} and \dataset{facebook}, which reduces the number of structural changes upon update.

In support to our hypothesis, for \dataset{amazon} we observe a low median speedup of $\approx 4$.
The low performance is due to the structural properties of this graph: very low clustering coefficient and high diameter lead to many structural changes upon edge addition or removal, and thus higher computational load.
For example, in both addition and removal of edges, we observe that on \dataset{dblp}, which is of the same order of magnitude as \dataset{amazon} but with a much higher clustering coefficient, our method achieves about double the speedup than on \dataset{amazon}.
We conjecture that performance on a larger graph is tightly connected with its structural properties, the longer disk access time, and more computational load per source (more vertices to be traversed).
Exploring the connections between algorithm's performance and graph properties is an interesting path for future investigation.

\subsection{Scalability for online updates}

\begin{figure}[t]
\begin{center}
	\includegraphics[scale=0.7]{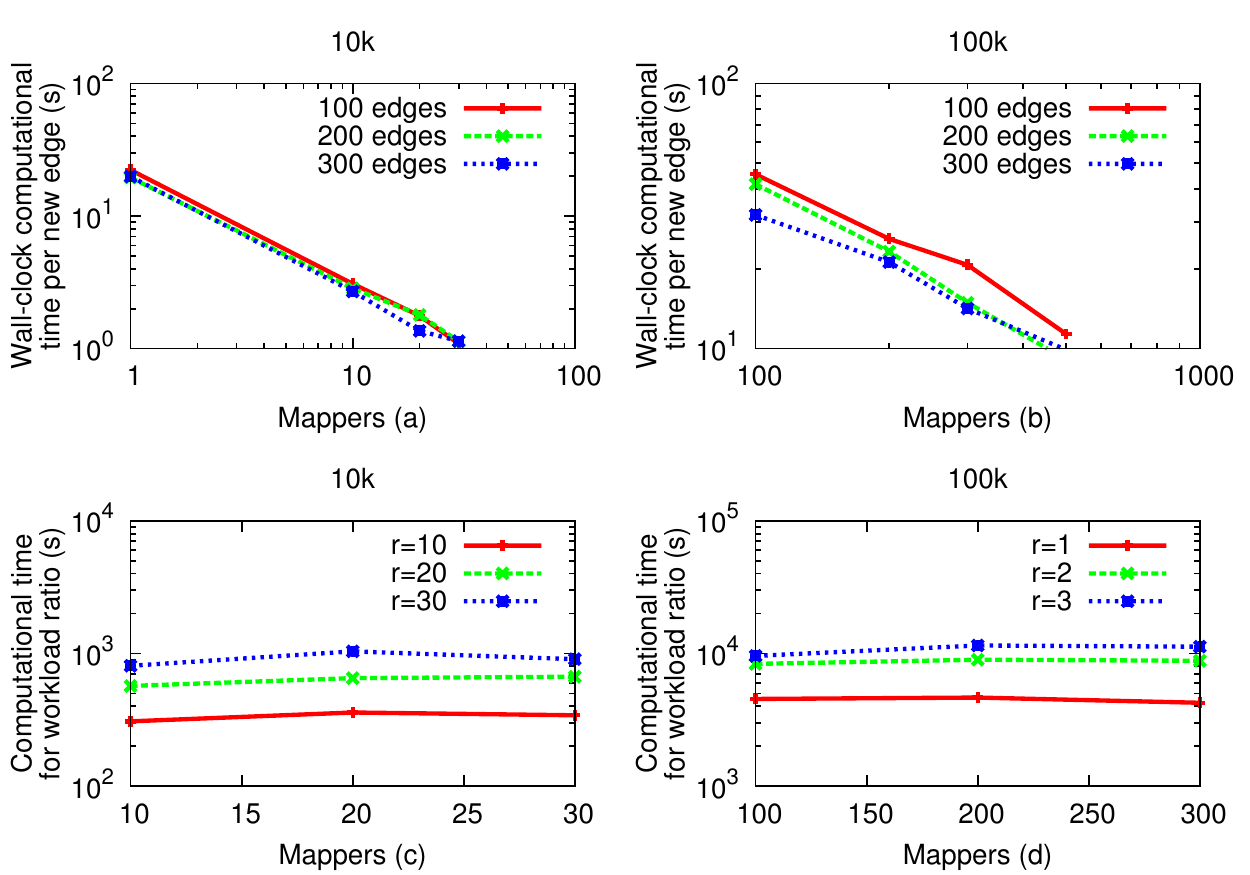}
\caption{(a-b) Computation time under increasing number of mappers.
(c-d) Computation time under constant ratio of workload over mappers.}
\label{fig:scaling_mr}
	\includegraphics[scale=0.7]{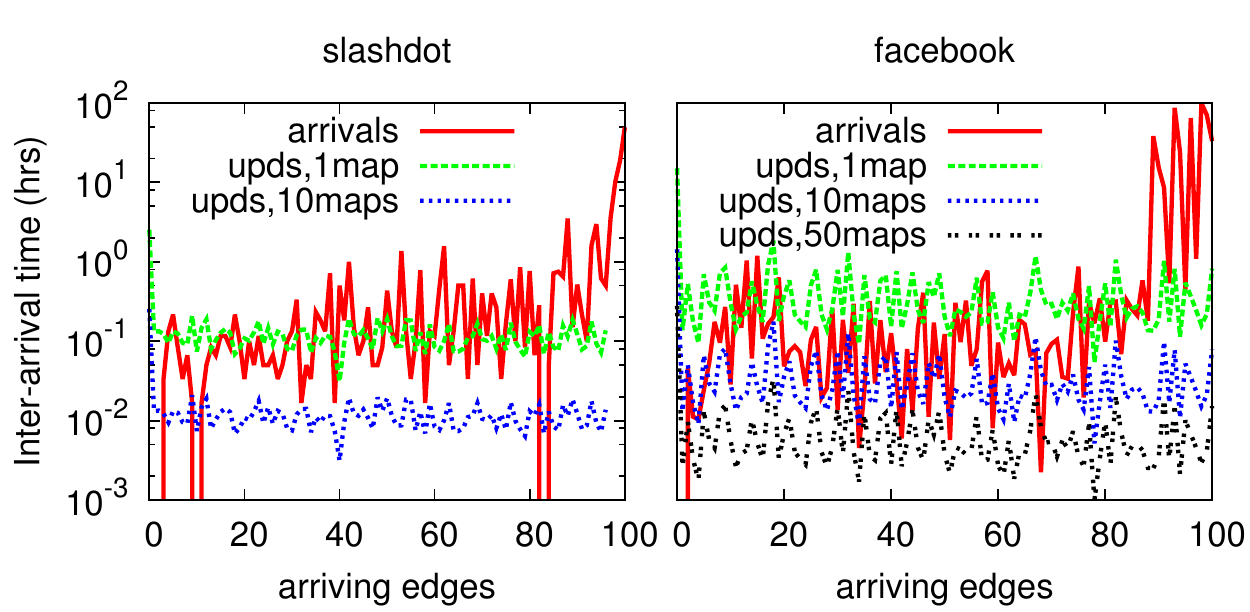}
\caption{Inter-arrival time of \dataset{facebook} and \dataset{slashdot} edges and update times for betweenness centrality.}
\label{fig:updates-interarrivals}
\end{center}
\end{figure}
Figures~\ref{fig:scaling_mr}(a-b) analyze the strong scaling properties of the algorithm in the case of edge addition.
In these experiments we keep the workload fixed and increase the parallelism level.
By employing a larger number of mappers, the overall execution time decreases almost linearly regardless of the workload (i.e., $100$ vs. $300$ added edges) and graph size (i.e., $10$k vs. $100$k).
As expected, our algorithm is embarrassingly parallel and shows very good scalability properties.

Figures~\ref{fig:scaling_mr}(c-d) explore the weak scaling properties of the algorithm in the case of edge addition.
In these experiments we keep the workload per processing unit fixed, and we increase the workload and parallelism level proportionally.
The workload for the system is represented by the number of edges updates within a period of time.
To keep the ratio of workload per mapper constant, we increase the number of edges updates proportionally to the number of mappers.
As shown in the figures, the total execution time remains constant at different levels of ratio (e.g., $1$ vs $3$, or $10$ vs. $30$) regardless of the parallelism level.
These results show that our parallel algorithmic framework can scale to larger workload simply by adding more machines.

Figure~\ref{fig:updates-interarrivals} demonstrates the \emph{online} capabilities of the algorithm on two real graphs.
The figure shows the inter-arrival time of new edges and the time needed by the framework to produce updated betweenness values.
In Table~\ref{tab:percents} we report the fraction of edges for which the framework was unable to produce updates on time and the corresponding average delay.
For \dataset{slashdot}, the framework manages to produce online updates for $98.91\%$ of edges with 10 mappers.
For \dataset{facebook} instead, the arrival rate is higher and $10$ machines are not enough.
However, the scalability of our algorithm allows to simply use more machines to decrease the response time of the system.
Thus, with 100 mappers, the system can produce online updates for $98.99\%$ of edges.

\begin{table}[t]
\caption{Edges missed and average delay vs. scaling.}
\centering
\small
\begin{tabular}{lrSS}
\toprule
Dataset			&	{mappers}	&	{\% missed}	&	{avg. delay (s)}	\\
\midrule
\dataset{slashdot}	&	1		&	44.565		&	257.9	\\
\dataset{slashdot}	&	10		&	1.087		&	32.4		\\ \hline
\dataset{facebook}	&	1		&	69.697		&	1061.1	\\
\dataset{facebook}	&	10		&	19.192		&	96.6		\\
\dataset{facebook}	&	50		&	3.030		&	8.6		\\
\dataset{facebook}	&	100		&	1.010		&	5.5		\\
\bottomrule
\end{tabular}
\label{tab:percents}
\end{table}

\subsection{Use-case: Girvan-Newman community detection}\label{sec:gn}

Our framework allows a faster update of the betweenness of edges, thus enabling several applications including community detection with the Girvan-Newman method~\cite{girvan02community}.
This method relies primarily on the computation of the betweenness of all edges.
It iteratively removes the edge with the highest centrality to create disconnected components.
The procedure is repeated until no edges remain, thus enabling the construction of a hierarchy of communities.
In practice, this algorithm has been abandoned due to the high cost incurred from the recalculation of the betweenness on the modified graph.
However, our framework allows a faster update of betweenness by taking into account only the affected parts of the graph, thus attaining an order of magnitude faster execution over Brandes when running the Girvan-Newman algorithm on a single machine (Figure~\ref{fig:girvan-newman}).

\begin{figure}[t]
\begin{center}
	\includegraphics[width=3.5in]{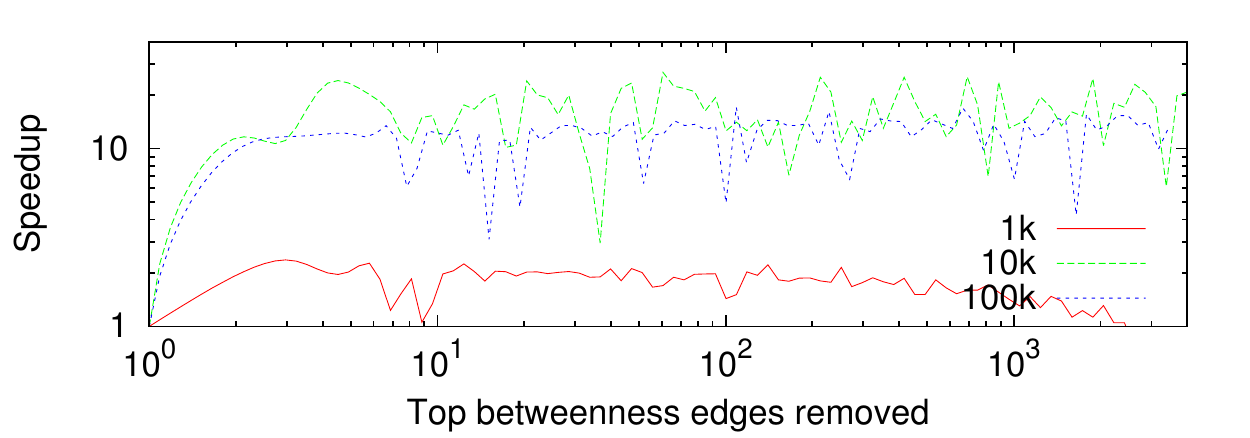}
\caption{Girvan-Newman: continuous removal of edge with highest betweenness and re-computation on synthetic graphs.}
\label{fig:girvan-newman}
\end{center}
\end{figure}

\section{Conclusions}\label{sec:discussion}

The computational complexity of most existing graph algorithms makes them impractical in nowadays massive and dynamic networks.
In order to scale graph analysis to real-world applications and to keep up with their highly dynamic nature, we need to devise new approaches specifically tailored for modern parallel stream processing engines that run on clusters of shared-nothing commodity hardware.

In this paper we introduce an algorithmic framework for computing betweenness centrality incrementally in large evolving graphs where edges and vertices are added and removed.
Our experimental results demonstrate that our framework is capable of scaling-out to a large number of machines without additional overhead.
This feature allows it to process graphs whose size is orders of magnitude larger than what previously reported in the literature.

Moreover, our method shows very good scale-up properties, which lead to an almost linear decrease of the execution time needed to update the betweenness centrality on parallel systems. As a result, our method is able to keep up with the incoming rate of updates in large real-world graphs and in an online fashion.

The scalability achieved by our framework opens the doors to new applications for real-world networks.
For instance, our framework can be exploited for online detection and prediction of emerging leaders and communities in social networks.



\end{document}